\documentclass[10pt]{article}
\usepackage{basepreamble}
\usepackage{preamble_final}
\usepackage{pxfonts}

\input xy
\xyoption{all}

\title{Sequential algorithms and the computational content of classical proofs}
\author{Thomas Powell}
\date{Preprint, \today}
\begin{document}

\maketitle

\begin{abstract}
We develop a correspondence between the theory of sequential algorithms and classical reasoning, via Kreisel's no-counterexample interpretation. Our framework views realizers of the no-counterexample interpretation as dynamic processes which interact with an oracle, and allows these processes to be modelled at any given level of abstraction. We discuss general constructions on algorithms which represent specific patterns which often appear in classical reasoning, and in particular, we develop a computational interpretation of the rule of dependent choice which is phrased purely on the level of algorithms, giving us a clearer insight into the computational meaning of proofs in classical analysis.
\end{abstract}

%%%%%%%%%%%%%%%%%%%%%%%%%%%%%%%%%%%%%%%%%%%%%%%%%
%%%%%%%%%%%%%%%%%%%%%%%%%%%%%%%%%%%%%%%%%%%%%%%%%
\section{Introduction}
%%%%%%%%%%%%%%%%%%%%%%%%%%%%%%%%%%%%%%%%%%%%%%%%%
%%%%%%%%%%%%%%%%%%%%%%%%%%%%%%%%%%%%%%%%%%%%%%%%%
\label{sec-intro}

The extraction of programs from proofs is an old and well developed concept. Ever since Hilbert's program, the connection between logic and computation has been a cornerstone of research in the foundations of mathematics and theoretical computer science, which has led to the construction of a large variety of formal interpretations that map proofs in a given system to terms which represent their computational content. The no-counterexample interpretation \cite{Kreisel(1951.0),Kreisel(1952.0)}, Dialectica interpretation \cite{Goedel(1958.0)} and modified realizability \cite{Kreisel(1959.0)} are examples of these. 

However, while interpretations give us a general procedure for extracting terms from proofs, actually understanding how those terms behave as \emph{algorithms} is an entirely different matter. The reason for this is clear: since the structure of extracted terms reflects the formal structure of the proof from which they were obtained, their salient features as programs lie hidden under a layer of syntax which accounts for every tiny logical step in the proof. Even the simplest proofs can yield programs whose behaviour is not at all obvious at first glance (see the study of the drinkers paradox carried out in \cite{Powell(2018.2)}), and those which make use of non-trivial classical principles such as the axiom of choice can be enormous (one of the first programs \cite{Murthy(1990.0)} extracted from a proof of Higman's lemma took up 12MB of computer code!).

More recent developments in the theory of program extraction have sought to address this problem. Classical realizability \cite{Krivine(2009.0)}, for example, makes use of a machine model which in particular allows a direct interpretation of classical reasoning via control operators, thus giving a more algorithmic account of programs obtained from classical proofs. In addition, refinements to traditional techniques have led to the extraction of much more efficient programs from proofs: see for instance \cite{BergBS(2002.0)} or \cite{Raffalli(2004.0)}. Nevertheless, these approaches are all fundamentally \emph{formal}, and involve a direct translation from proofs to terms.

In this paper, we take a different approach. We focus on the production not of terms in some formal language, but of \emph{high-level descriptions} of algorithms. In doing so, we propose a technique of program extraction which combines both formal and informal reasoning: appealing to human intuition in conjunction with powerful formal constructions on algorithms which are inspired by proof-theoretic techniques.

To be more precise, we develop a simple and very general kind of oracle transition system, inspired by Gurevich's notion of a sequential time algorithm \cite{Gurevich(2000.0)}. This models a class of algorithms which operate on some underlying state, and allow explicit queries to some `mathematical environment'. We set up a correspondence between sequential algorithms and Kreisel's famous no-counterexample interpretation, whereby an algorithm represents a realizer of the n.c.i. of some formula iff elements of its end state satisfy the quantifier-free part of that formula. Though in this paper we work primarily with the n.c.i., together with aspects of the Dialectica interpretation, the basic idea could be adapted to other settings.

In our main technical result, we give a construction on algorithms which interprets the rule of dependent choice. We consider this result to be of interest in its own right: Though inspired by Spector's bar recursor \cite{Spector(1962.0)}, our proof of correctness does not appeal to any sort of choice or bar induction on the metalevel, nor to any semantic properties such as continuity. Moreover, it gives an explicit step-by-step construction of the underlying choice sequence, something which is typically obscured in traditional approaches.

However, our primary motive is practical, and based on the following idea: Many non-constructive proofs of existential statements in mathematics contain a single instance of a strong non-constructive principle, such as choice or Zorn's lemma, and are otherwise computationally simple. Provided we have a clear algorithmic interpretation of the strong principle, we can often convert this to an algorithm for computing a witness to the theorem. We present a framework in which this idea can be made precise.

As such, the content of this paper should be viewed on two levels. On the one hand, we give a new computational interpretation of dependent choice using sequential algorithms, which provides us with some new insight into the computational content of classical analysis, and could be used to better understand bar recursive programs as extracted in e.g. \cite{OliPow(2015.1),OliPow(2015.0)}. On the other, this can all be viewed as a single illustration of a much more general approach to program extraction, which focuses on the intuitive representation of programs as algorithms.

It is important to emphasise that we do not consider ideas of this paper as replacing traditional methods of program extraction such as proof interpretation! General techniques which act on formalised proofs are essential if our goal is to e.g. synthesise verified software. Our approach is more in line with that of dynamical algebra \cite{CosLomRoy(2001.0)}, in which clever algorithms are sought to replace the use of e.g. a maximal ideal in some existence proof, as in \cite{Yengui(2008.0)}. Our paper proposes an approach for designing algorithms of this kind which is more closely related to traditional proof interpretations, and may eventually lead to a better understanding of the connection between dynamical algebra and techniques such as the n.c.i. and the Dialectica interpretation.

%Proof interpretations such as the n.c.i. or the Dialectica were initially developed for foundational purposes, for relative consistency proofs or characterising the provable recursive functions of some theory. For applications of this kind, an algorithmic understanding of extracted programs is irrelevant. Moreover, 

%%%%%%%%%%%%%%%%%%%%%%%%%%%%%%%%%%%%%%%%%%%%%%%%%
\subsection{A remark on terminology}
%%%%%%%%%%%%%%%%%%%%%%%%%%%%%%%%%%%%%%%%%%%%%%%%%
\label{sec-intro-term}

Our use of the term 'sequential algorithm' follows Gurevich \cite{Gurevich(2000.0)}, although our notion of an algorithm is much simpler, and might be more accurately described as an 'oracle transition system'. Our aim is simply to give a flexible and above all high-level characterisation of algorithms that can be viewed as some evolving computation on a state. In particular, our focus is quite different from sequential algorithms \cite{BerCur(1982.0)} or dialog games \cite{HylOng(2000.0)} as seen in the context of the semantics of PCF. 

%%%%%%%%%%%%%%%%%%%%%%%%%%%%%%%%%%%%%%%%%%%%%%%%%
\subsection{Outline of paper}
%%%%%%%%%%%%%%%%%%%%%%%%%%%%%%%%%%%%%%%%%%%%%%%%%
\label{sec-intro-overview}

Section \ref{sec-nci} provides a brief overview of the no-counterexample interpretation. Sections \ref{sec-alg}-\ref{sec-thmalg} set up our framework, and in particular our notion of an oracle sequential algorithm. We present our main technical result on algorithms - an interpretation of the rule of dependent choice - in Section \ref{sec-dc}, and conclude, in Section \ref{sec-tape}, with a detailed case study, in which we consider the so-called `infinite tape' example.

%%%%%%%%%%%%%%%%%%%%%%%%%%%%%%%%%%%%%%%%%%%%%%%%%
%%%%%%%%%%%%%%%%%%%%%%%%%%%%%%%%%%%%%%%%%%%%%%%%%
\section{The no-counterexample interpretation}
%%%%%%%%%%%%%%%%%%%%%%%%%%%%%%%%%%%%%%%%%%%%%%%%%
%%%%%%%%%%%%%%%%%%%%%%%%%%%%%%%%%%%%%%%%%%%%%%%%%
\label{sec-nci}

We begin by introducing, in a rather informal setting, Kreisel's famous no-counterexample interpretation (henceforth abbreviated as the n.c.i). It should be mentioned that while the n.c.i. is closely related to the combination of the negative translation and G\"{o}del's functional (or Dialectica) interpretation, in the case of $\forall\exists\forall$-formulas, which we study here, the two coincide, and though we use ideas from the Dialectica interpretation throughout the paper, we explain these as and when we need them, and so a more detailed introduction to proof interpretations and higher-order computation is not necessary.

Unlike most works on the subject of proof interpretations, which fix a formal type theory from the outset, we work in a simple mathematical setting of sets $X,Y,Z,\ldots$ and functions $f:X\to Y$ between these sets. We will also talk about \emph{functionals} $\Phi:Y^X\to Z$ i.e. maps which take functions as inputs, and will later impose some conditions on functionals such as \emph{continuity}. Usually all this is formalised in some higher-order extension of arithmetic such as the $\PAw$ of \cite{Troelstra(1973.0)}, which talk more generally about objects of arbitrary finite type. However, because we are concerned with high-level algorithms rather than detailed general soundness theorems, an informal setting works well and allows us a lot more flexibility.

Now, suppose we have a theorem $A$ of the form 
\begin{equation*}
A:\equiv (\forall u \in U)(\exists x \in X)(\forall y \in Y) P(u,x,y).
\end{equation*}
where the predicate $P(u,x,y)$ is taken to be decidable, which means that we assume some encoding of objects $u,x,y$ and a computable procedure which determines the truth value of $P(u,x,y)$. 

In general, it is not possible to find a direct witness for the existential quantifier $\exists x$. The usual example is given by the Kleene $T$-predicate:
\begin{equation*}
(\forall u\in\NN)(\exists x\in\NN)(\forall y\in\NN)(T(u,u,x)\vee \neg T(u,u,y)).
\end{equation*}
Here, if we could compute a bound on the $\exists x$ in terms of $u$, we would have solved the halting problem. However, it is possible to reformulate such statements so that they become computationally tractable. Let's suppose that the theorem $A$ above is not true i.e. there is some $u\in\ U$ such that
\begin{equation*}
(\forall x \in X)(\exists y \in Y) \neg P(u,x,y).
\end{equation*}
Then (by the axiom of choice for quantifier-free formulas) there would exist for some $u$ a `counterexample function' $f:X\to Y$ satisfying
\begin{equation*}
(\forall x) \neg P(u,x,f(x))
\end{equation*}
Therefore, negating a second time, we see that $A$ is equivalent to the statement $\nci{A}$ that \emph{any counterexample function must fail}, or in logical terms:
\begin{equation*}
\nci{A}:\equiv (\forall u\in U,\forall f:X\to Y)(\exists x\in X) P(u,x,f(x)).
\end{equation*}
An explicit witness for the quantifier $\exists x$ would then be given by a functional $\Phi:U\times Y^X\to X$ which satisfies
\begin{equation*}
(\forall u,f:X\to Y) P(u,\Phi(u,f),f(\Phi(u,f))).
\end{equation*}
It turns out that in very general circumstances it is possible to find such a functional $\Phi$ which is \emph{computable} in the oracle $f$. In our simple example of the Halting problem, the n.c.i.-version of the theorem is just
\begin{equation*}
(\forall u,f:\NN\to\NN)(\exists x) (T(u,u,x)\vee \neg T(u,u,f(x)))
\end{equation*}
and this is realized by the functional defined as
\begin{equation*}
\Phi(u,f):=\begin{cases}0 & \mbox{if $\neg T(u,u,f(0))$}\\ f(0) & \mbox{otherwise}.\end{cases}
\end{equation*}
which can easily be computed by some suitable oracle Turing machine. 

Typically this phenomenon is made precise through a soundness theorem, which states that whenever $A$ is provable in some theory $\mathcal{T}$ then there is a corresponding $\Phi$ in some subrecursive class of functionals $\mathcal{F}$. In the case that $\mathcal{T}$ is Peano arithmetic, a realizer for the n.c.i. of $A$ can be found in G\"{o}del's System T of primitive recursive functionals of all finite type:
\begin{theorem}[Kreisel]
Suppose that $(\forall u)(\exists x)(\forall y) P(u,x,y)$ is provable in Peano arithmetic, where $P(u,x,y)$ is a quantifier-free formula. Then there is a term $t$ of System T such that $(\forall u,f) P(u,tuf,f(tuf))$ is provable in System T. 
\end{theorem} 

\begin{remark}
A proof of the theorem in its form here is given in e.g. \cite[Proposition 10.9]{Kohlenbach(2008.0)}. However, Kreisel initially proved that the n.c.i. could be realized by some $\alpha$-recursive functional $\Phi$ for some $\alpha<\varepsilon_0$, which follows by applying a normalization procedure to the term $t$ in the above result. Naturally, the theorem holds more generally for formulas of the form 
\begin{equation*}
(\exists x_1)(\forall y_1)\ldots (\exists x_n)(\forall y_n) P(x_1,y_1,\ldots,x_n,y_n)
\end{equation*}
which would be interpreted by a sequence of functionals $\Phi_1,\ldots,\Phi_n$ satisfying
\begin{equation*}
(\forall f_1,\ldots,f_n) P(\Phi_1(\vec{f}),f_1(\Phi_1(\vec{f})),\ldots,\Phi_n(\vec{f}),f_n(\Phi_1(\vec{f}),\ldots,\Phi_n(\vec{f})))
\end{equation*}
but in this article we stick to at most three quantifier alternations.
\end{remark}
The n.c.i. gives us a connection between theorems of the form $A\equiv (\forall u\in U)(\exists x\in X)(\forall y\in Y) P(u,x,y)$ and computable functionals of type $U\times Y^X\to X$. There are several ways in which the behaviour of such functionals can be characterised. One which is particularly illuminating is the idea that, while $A$ states that an ideal object $x$ exists, which satisfies $P(u,x,y)$ for all $y$, $\nci{A}$ states that for any $f$ an \emph{approximation} to $x$ exists, which satisfies $P(u,x,y)$ only for $y=f(x)$. In this sense, $\Phi$ is a functional for computing approximations to non-computable objects, and in practise, it is often the case that terms representing $\Phi$ carry out some kind of learning, or trial-and-error procedure. The n.c.i. has been studied from this perspective in a number of places, see for example \cite{KohSaf(2014.0),Krivine(2003.0),Powell(2016.0)}.

In this article, we focus on the relationship between theorems and \emph{sequential algorithms}. Instead of focusing on general soundness theorems which formally produce e.g. $\alpha$-recursive functionals or closed terms of System T, we establish a connection between theorems and a kind of oracle transition system which we will call an \emph{approximation algorithm}:
\begin{equation*}
\mbox{Theorem}\stackrel{\scriptsize{n.c.i.}}{\leftrightarrow}\mbox{Approximation algorithm}
\end{equation*}
The main results of the paper will be a handful of results which explore this connection, and in particular establish constructions which act directly on algorithms.

%%%%%%%%%%%%%%%%%%%%%%%%%%%%%%%%%%%%%%%%%%%%%%%%%
%%%%%%%%%%%%%%%%%%%%%%%%%%%%%%%%%%%%%%%%%%%%%%%%%
\section{Approximation algorithms}
%%%%%%%%%%%%%%%%%%%%%%%%%%%%%%%%%%%%%%%%%%%%%%%%%
%%%%%%%%%%%%%%%%%%%%%%%%%%%%%%%%%%%%%%%%%%%%%%%%%
\label{sec-alg} 

We begin by introducing, and carefully motivating, our notion of an approximation algorithm. While this is essentially nothing more than a transition system with oracle queries, there are several details and design choices that need to be made. We start by outlining the very simple notion of a sequential time algorithm.
\begin{definition}
\label{def-sa}
A sequential algorithm $\A:=(S,I,E,\rhd)$ consists of:
\begin{itemize}

\item A set $S$ of states.

\item Subsets $I,E\subset S$ of initial and end states.

\item A partial function $\rhd:S\to S$.

\end{itemize}
We write $t=\rhd(s)$ inline as $s\rhd t$.
\end{definition}

\begin{definition}
\label{def-run}
A \emph{run} in $\A$ is a sequence
\begin{equation*}
s_0\rhd s_1\rhd\ldots \rhd s_n
\end{equation*}
where we assume that $s_i\notin E$ for $i<n$. We write $s_0\rhd^\ast s_n$ whenever there is a run from $s_0$ to $s_n$. The element $t$ is \emph{reachable} from $s$ if $s\rhd^\ast t$.
\end{definition}

\begin{definition}
\label{def-term}
We say that $s$ \emph{terminates} and write $\wf{s}$ if there exists some $t\in E$ such that $s\rhd^\ast t$. By our definition of a run, $t$ must be the first element of $E$ reachable from $s$. We say that $\A$ terminates if $\wf{s}$ for all $s\in I$.
\end{definition}

So far, our definition essentially coincides with the notion of a sequential time algorithm in Gurevich \cite{Gurevich(2000.0)} (cf. Postulate 1), and is of course nothing more than a simple deterministic transition system. 

This already allows us to give clean, high level descriptions of sequential processes: For instance, the Euclidean algorithm can be described as the sequential algorithm $\A$ with $S:=\NN\times\NN$, $I:=\{(x,y)\; | \; x>y\}$ and $E:=\{(x,0)\; | \; x\in\NN\}$, with transitions given by
\begin{equation*}
(x,y)\rhd (y,x\;\mathrm{ mod }\; y).
\end{equation*}
While the algorithms of Gurevich obey further postulates e.g. states are structures to which only local changes are made and so on, we do not refine our definition any further: Our priority here is not on a complete characterisation of a sequential algorithm, but on devising a simple abstract model through which the computational behaviour of certain non-constructive principles can be clearly understood. However, we conjecture that were our definition to be enhanced with further conditions along the lines of \cite{Gurevich(2000.0)}, most of what follows could be adapted to the more restrictive setting. 

%%%%%%%%%%%%%%%%%%%%%%%%%%%%%%%%%%%%%%%%%%%%%%%%%
\subsection{Input/output}
%%%%%%%%%%%%%%%%%%%%%%%%%%%%%%%%%%%%%%%%%%%%%%%%%
\label{sec-alg-oi} 

In order to make explicit the idea that a sequential algorithm $\A$ computes a function $f:U\to V$, one needs to add some additional structure dealing with input and output. We incorporate this into the following expanded definition, which we call an \emph{explicit} sequential algorithm.

\begin{definition}
\label{def-esa}
An explicit sequential algorithm $\A:=(S,E,\rho,\pi,\rhd)$ of sort $U,V$ is a sequential algorithm which comes equipped with functions $\rho:U\to S$ and $\pi:S\to V$, where input states are now a secondary notion defined as $I:=\rho(U)$.
\end{definition}
We could have naturally restricted the domain of $\pi$ to be $E$ rather than $S$. However, throughout this paper, $\pi$ should be interpreted as a \emph{location} in the state, which will typically just be a projection, and as such, it just returns whatever is currently in the output location, regardless of whether the state is an end state or not. Moreover, algorithms associated to proofs are typically characterised as `building' some witness $v\in V$ to an existential statement by trial and error. Thus we want this to be explicitly part of the evolving state.
\begin{definition}
\label{def-indfun}
An explicit sequential algorithm $\A$ of sort $U,V$ induces a partial function $f_\A:U\to V$ in the obvious way:
\begin{equation*}
f_\A(u):=\begin{cases}\pi(t) & \mbox{if $\rho(u)\rhd^\ast t\in E$}\\
\mbox{undefined} & \mbox{otherwise} \end{cases}
\end{equation*}
We say, conversely, that a partial function $f:U\to V$ is \emph{computed} by a sequential algorithm $\A$ if $f_\A(u)=f(u)$ for all $u\in U$. 
\end{definition}

\begin{proposition}
\label{prop-term}
If $\A$ terminates then the partial function $f_\A$ is total.
\end{proposition}

\begin{example}
The Euclidean algorithm can be characterised as an explicit sequential algorithm $\A$ of sort $\NN\times\NN$, $\NN$ by defining $S:=\NN\times\NN$, $\rho(a,b):=(\max(a,b),\min(a,b))$, $\pi(x,y):=x$, $E:=\{(x,0)\; | \; x\in\NN\}$ and $(x,y)\rhd (y,x\; \mathrm{mod}\; y)$. The induced function $f_\A$ returns the gcd of any pair of natural numbers: for example, $f(28,72)$ would be computed as
\begin{equation*}
(72,28)\rhd (28,16)\rhd (16,12)\rhd (12,4)\rhd (4,0)
\end{equation*}
yielding $f(28,72)=\pi(4,0)=4$, and the sequence $72,28,16,12,4$ can be seen as a succession of improving approximations to the gcd. 
\end{example}

Note that our notion of an algorithm is extremely broad. In particular, we make no assumptions about $\rhd$: as such, any function $f:U\to V$ is trivially computed by the algorithm $\A:=(\{0,1\}\times U\times V,\rho,\pi,E,\rhd)$ of sort $U,V$, where $\rho(u)=(0,u,v_0)$ for some $v_0\in V$, $\pi(b,u,v):=v$, end states are those of form $(1,u,v)$ and $(0,u,v)\rhd (1,u,f(u))$. 

Naturally, though, our intention is that when a function is induced by some algorithm, the latter gives some meaningful, though potentially high-level, description of how that function is \emph{computed}. We could make the notion of a computable algorithm precise as follows:
\begin{definition}
\label{def-comp}
An explicit algorithm $\A:=(S,E,\rho,\pi,\rhd)$ of sort $U,V$ is computable if $\rho,\pi$ and $\rhd$ are computable (resp. partial computable) functions relative to some representations of $U,V$ and $S$, and membership of $E$ is decidable.
\end{definition}

\begin{proposition}
\label{prop-comp}
If $\A$ is computable then the function $f_\A$ is partial computable.
\end{proposition}

\begin{proof}
We sketch the appropriate Turing machine: given some input $u$, we first run a Turing machine which computes $\rho$, so that the initial state $s_0:=\rho(u)$ is stored somewhere in the memory. We now simulate the algorithm itself. We first run a Turing machine for $E$ to check whether or not $s_0\in E$. If it is, then we run $\pi$ on input $s_0$ and halt. Otherwise, it run a Turing machine for $\rhd$ to compute $s_1:=\rhd(s_0)$. We then repeat the above process. If $f_\A(u)$ is defined, then there is some $t\in E$ such that $\rho(u)\rhd^\ast t$. In which case, the Turing machine will eventually halt with $\pi(t)=f(u)$ in its memory.
\end{proof}

Computability issues do not play an important role in this paper. All of the algorithms we discuss are computable, and all constructions on algorithms will clearly preserve computability, and we do not try to make this precise. For us, sequential algorithms are a way of \emph{describing} functions, rather than a characterisation of their computational strength.

%We will now enrich our basic definition in a certain way, to make input/output slightly less cumbersome, and to incorporate the notion of an oracle query. We deal with each of these in turn.

%%%%%%%%%%%%%%%%%%%%%%%%%%%%%%%%%%%%%%%%%%%%%%%%%
\subsection{Oracle queries}
%%%%%%%%%%%%%%%%%%%%%%%%%%%%%%%%%%%%%%%%%%%%%%%%%
\label{sec-alg-app} 

Explicit sequential algorithms in the usual sense allow us to define functions. However, this paper is primarily concerned with \emph{functionals}, which means that we must now extend our definition of an algorithm to allow for oracle queries. Technically speaking, we could accomplish this by including an oracle as some parameter for our transition function $\rhd$. However, it will be more elegant and meaningful to include oracles explicitly as part of the definition of an algorithm. The following is a simple modification of Definition \ref{def-esa}, and constitutes our formulation of an oracle sequential algorithm:
\begin{definition}
\label{def-osa}
An oracle sequential algorithm (OSA) of sort $U,V,X,Y$ is given by a tuple $\A:=(R,Q,E,\rho,\xi,\pi,\rhd)$, whose components are as follows:
\begin{itemize}

\item  $R$ is a set of registers (which were formerly our states), while states are now defined by $S:=R\times Y_\Box$ where $Y_\Box:=Y+\{\Box\}$ for some object $\Box$ which denotes an `empty register'. States will be written as 
\begin{equation*}
\stac{r}{o}\mbox{ where $o=\Box$ or $o\in Y$},
\end{equation*}
and we define $S_\Box:=\{\stac{r}{\Box}\; | \; r\in R\}$;

\item $Q\subseteq S_\Box$ are \emph{query} states; 

\item $E\subseteq S$ are end states, where we assume that $Q\cap E=\emptyset$;

\item $\rho:U\to R$ is an input map and we define $I:=\{\rho_u\; | \; u\in U\}\subseteq S_\Box$ where $\rho_u:=\stac{\rho(u)}{\Box}$;

\item $\xi:R\to X$ returns the current query; 

\item $\pi:S\to V$ is an output map;

\item $\rhd:S\to S$ is a partial function.

\end{itemize}

\end{definition}

States of an OSA now come equipped with a `answer tape', which is either in an `empty' state $\Box$ or contains some $y\in Y$. They can also now be in a query position, in which case they query the oracle on some value $\xi(r)$ stored in $r$, which writes the answer in the answer tape. We make this formal as follows:

\begin{definition}
\label{def-oracle}
An oracle for an OSA $\A$ of sort $X,Y,U,V$ is a function $f:X\to Y$. For any oracle, there is a corresponding oracle transition $\orhd_f:Q\to S\backslash S_\Box$ defined by
\begin{equation*}
\stac{r}{\Box}\orhd_f \stac{r}{y}\mbox{ for $y:=f(\xi(r))$}.
\end{equation*}
\end{definition}

The role of the empty element $\Box$ is to differentiate `before' and `after' the oracle transition. Of course, in a concrete computational model of our OSA, a transition to an empty oracle tape need not mean deleting the contents of the tape: it just refers to a marker which indicates that the contents of the tape are no longer relevant, since we are anticipating a new oracle answer. 

\begin{definition}
\label{def-orun}
A run in an OSA $\A$ on the oracle $f$ is a sequence of states
\begin{equation*}
s_0,s_1,\ldots,s_n
\end{equation*}
where
\begin{itemize}

\item $s_i\notin E$ for $i<n$,

\item if $s_i\notin Q$ then $s_i\rhd s_{i+1}$,

\item if $s_i\in Q$ then $s_i\orhd_f s_{i+1}$.

\end{itemize}
We write $s\rhd^\ast_f t$ whenever there is a run on $f$ from $s=$ to $t$, and say that $t$ is reachable from $s$ on $f$.
\end{definition}

Note that our notion of a run is similar to the \emph{interactive runs} of Gurevich \cite[Section 8]{Gurevich(2000.0)}. The following definitions and results all lift directly to oracle sequential algorithms.

\begin{definition}
\label{def-oterm}
We say that $s$ terminates on $f$, and write $\wfs{s}{f}$, if $s\rhd_f^\ast t\in E$. An OSA terminates on some class $\F$ of oracles if $\wfs{\rho_u}{f}$ for any $f\in \F$ and all $u\in U$.
\end{definition}

We will often suppress $f$ on both $\rhd_f$ and $\wfs{s}{f}$ when there is no risk of ambiguity.

\begin{definition}
\label{def-indfunc}
An OSA of sort $X,Y,U,V$ induces a partial functional $\Phi_\A:U\times Y^X\to V$ as follows:
\begin{equation*}
\Phi_\A(u,f):=\begin{cases}\pi(t) & \mbox{if $\rho_u\rhd_f^\ast t\in E$}\\
\mbox{undefined} & \mbox{otherwise} \end{cases}
\end{equation*}
\end{definition}

\begin{proposition}
\label{prop-oterm}
If $\A$ terminates on $\F$ then the partial functional $\Phi_\A$ is total when restricted to $U\times \F\subseteq U\times Y^X$.
\end{proposition}

\begin{definition}
\label{def-ocomp}
An OSA of sort $X,Y,U,V$ is computable if $\rho,\xi,\pi$ and $\rhd$ are computable (resp. partial computable) functions, and membership of both $Q$ and $E$ is decidable, relative to some representations of $X,Y,U,V$.
\end{definition}

\begin{proposition}
\label{prop-ocomp}
If $\A$ is computable then the functional $\Phi_\A$ is computable (i.e. can be computed via an oracle Turing machine).
\end{proposition}

\begin{proof}
A simple adaptation of the proof of \ref{prop-comp}. As before, after initialising the input we simulate the run on $\rhd_f$. Now, in addition, at each state $\stac{r}{\Box}$ the Turing machine needs to test whether or not $\stac{r}{\Box}\in Q$. If it is, then it computes $\xi(s)$ and writes it to the machine's oracle query tape, and after receiving some answer $y$ sets the current state to $\stac{r}{y}$. If $s\notin Q$ then it computes $s':=\rhd(s)$ as before. If $\Phi_\A(u,f)$ is defined there is some $t\in E$ such that $\rho_u\rhd^\ast t$, and then the Turing machine when run on oracle $f$ will halt with $\pi(t)=\Phi_\A(u,f)$ in its memory.
\end{proof}

\begin{definition}
\label{def-call}
Suppose that $\A$ is an approximation algorithm which terminates on $u$ and $f$ i.e. there is a run 
\begin{equation*}
\rho_u:=s_0, s_1,\ldots , s_n\in E
\end{equation*}
Let $s_{i_0},\ldots,s_{i_{m-1}}$ denote the subsequence of those $s_j\in Q$, and let $x_j:=\xi(r_{i_j})$ where $s_{i_j}=\stac{r_{i_j}}{\Box}$. We call the sequence $x_0,\ldots,x_{m-1}$ the \emph{query sequence} of $\A$ on $u$ and $f$.
\end{definition}

\begin{proposition}[Continuity]
\label{prop-cont}
Suppose that $\A$ is terminating on $\F$. Then for all $u$ and $f\in \F$ there exists a finite sequence $x_0,\ldots,x_{m-1}$ such that for all $g:X\to Y$
\begin{equation*}
\forall j<m(f(x_j)=g(x_j))\to \Phi_\A(u,f)=\Phi_\A(u,g).
\end{equation*}
\end{proposition}

\begin{proof}
Let $x_0,\ldots,x_{m-1}$ be the query sequence of $\A$ on $u$ and $f$. We prove by induction on $i$ that if $s_0,\ldots,s_i$ is a run of $\A$ on $u$ and $f$ and $s'_0,\ldots,s'_i$ a run of $\A$ on $g$, then $s_k=s'_k$ for all $k<i$. For $i=0$ this is trivial, since $s_0=\rho_u=s'_0$. Otherwise, assuming that $s_i\rhd_f s_{i+1}$ and $s_i=s'_i$, we know that $s_i\notin E$, so there are two possibilities. Either $s_i\notin Q$ and so $s'_{i+1}=\rhd(s'_i)=\rhd(s_i)=s_{i+1}$ or $s'_{i+1}=\orhd_g(s_i)=\orhd_f(s_i)=s_{i+1}$ which follows from $g(\xi[s_i])=g(x_j)=f(x_j)=f(\xi[s_i])$ for some $x_j$ in the query sequence. Therefore if $\rho_u\rhd^\ast_f t\in E$ then also $\rho_u\rhd^\ast_g t\in E$ and thus $\Phi_\A(u,f)=\pi(t)=\Phi_\A(u,g)$.
\end{proof}

\begin{example}
The functional $\Phi:\NN\times\NN^\NN\to\NN$ defined by $\Phi(n,f):=\max_{i\leq n}f(i)$ is computed by the following OSA: $S:=\NN\times\NN$, $Q:=\{\stac{i,n}{\Box}\; | \; n>0\}$, $E:=\{\stac{i,n}{\Box}\; | \; n=0\}$, $\rho(n):=(0,n+1)$, $\xi(i,n):=n-1$ and $\pi\stac{i,n}{o}:=i$, with $\stac{i,n+1}{y}\rhd \stac{\max(i,y),n}{\Box}$. For example, $\Phi(2,f)$ would be computed as
\begin{equation*}
\begin{aligned}
&\stac{0,3}{\Box}\orhd \stac{0,3}{y_0}\rhd \stac{y_0,2}{\Box}\orhd \stac{y_0,2}{y_1}\rhd \stac{\max(y_0,y_1),1}{\Box}\\
&\orhd  \stac{\max(y_0,y_1),1}{y_2}\rhd \stac{\max(y_0,y_1,y_2),0}{\Box}\in E
\end{aligned}
\end{equation*}
where $y_0,y_1,y_2=f(2),f(1),f(0)$.

\end{example}

%%%%%%%%%%%%%%%%%%%%%%%%%%%%%%%%%%%%%%%%%%%%%%%%%
%%%%%%%%%%%%%%%%%%%%%%%%%%%%%%%%%%%%%%%%%%%%%%%%%
\section{Theorems and algorithms}
%%%%%%%%%%%%%%%%%%%%%%%%%%%%%%%%%%%%%%%%%%%%%%%%%
%%%%%%%%%%%%%%%%%%%%%%%%%%%%%%%%%%%%%%%%%%%%%%%%%
\label{sec-thmalg} 

We have introduced Kreisel's no counterexample interpretation, and defined a notion of an oracle sequential algorithm which computes functionals of n.c.i. realizer type. We now make the connection between these concepts precise, and give a number of definitions and results which develop it further.

While it is clear that we can view an OSA of sort $U,X,X,Y$ as a realizer of the n.c.i. of some formula $(\forall u\in U)(\exists x\in X)(\forall y\in Y) P(u,x,y)$ if $\Phi_\A:U\times Y^X\to X$ satisfies $(\forall u\in U, f:X\to Y)P(u,\Phi_\A(u,f),f(\Phi_\A(u,f)))$, we give a refinement of an OSA which makes the connection slightly cleaner, and will be crucial in Section \ref{sec-dc}.

\begin{definition}
\label{def-aa}
Let $\A=(R,Q,E,\rho,\xi,\pi,\rhd)$ be an OSA of sort $U,X,X,Y$ where:
\begin{enumerate}[(i)]

\item\label{item-aai} $\pi\stac{r}{o}=\xi(r)$;

\item\label{item-aaii} $E\subseteq S\backslash S_\Box$;

\item\label{item-aaiii} if $\stac{r}{o}\rhd \stac{r'}{y}$ with $y\in Y$ then $o=y$ and $\xi(r)=\xi(r')$.

\end{enumerate}
Then we call $\A$ an \emph{approximation algorithm} of sort $U,X,Y$.

\end{definition}
Intuitively, the first condition (\ref{item-aai}) combines the query and output locations of the state, the second ensures that end states always have the form $\stac{r}{y}$ for $y\in Y$, and the third results in the following property:
\begin{lemma}
\label{lem-nci}
Suppose that $\rho_u\rhd^\ast_f \stac{r}{o}$ for $u\in U$. Then if $o\in Y$ then $o=f(\xi(r))$.
\end{lemma}

\begin{proof}
Induction on runs. The only run of length $1$ is the single element $\rho_u$ and so the result trivially holds since $\rho_u=\stac{\rho(u)}{\Box}$. Now suppose that $s,\ldots,s_i,s_{i+1}=\stac{r}{o}$ is a run. If $s_i\in Q$ then $s_i\orhd_f \stac{r}{o}$ and so $s_i=\stac{r}{\Box}$ and $o=f(\xi(r))$. Otherwise, $s_i\rhd \stac{r}{o}$ and either $o=\Box$, in which case there is nothing to prove, or $o=y\in Y$, in which case $s_i=\stac{r'}{y}$ with $\xi(r')=\xi(r)$ by the condition (\ref{item-aaiii}). Since we assume inductively that $y=f(\xi(r'))$, we have $y=f(\xi(r))$. 
\end{proof}

\begin{definition}
\label{def-pred}
Let $P(u,x,y)$ be a predicate on $U,X,Y$. An approximation algorithm $(R,Q,E,\rho,\xi,\rhd)$ of sort $U,X,Y$ \emph{satisfies} $P$ on $u\in U$ and $f:X\to Y$ if
\begin{equation*}
\rho_u\rhd_f^\ast \stac{r}{y}\in E\mbox{ \ \ \ implies \ \ \ }P(u,\xi(r),y).
\end{equation*}

\end{definition}

\begin{theorem}
\label{thm-nci}
Suppose that the approximation algorithm $\A$ satisfies $P(u,x,y)$ and terminates on $\F$. Then the induced functional $\Phi_\A$ satisfies the n.c.i. of $\forall u\exists x\forall y P(u,x,y)$ on $\F$.
\end{theorem}

\begin{proof}
Take $f\in \F$. Then $\rho_u\rhd^\ast_f \stac{r}{y}$ with $P(u,\xi(r),y)$. But $\Phi_\A(u,f):=\xi(r)$ and by Lemma \ref{lem-nci} we have $y=f(\xi(x))=f(\Phi_\A(u,f))$. The result follows.
\end{proof}

\begin{corollary}
If $\A$ is computable, realizes $P$ and terminates, then $\Phi_\A$ is a computable functional which satisfies the n.c.i. of $\forall u\exists x\forall y P(u,x,y)$. 
\end{corollary}

The obvious question now is how do we construct an appropriate algorithm $\A$ from a given proof? It is here that our approach departs from the standard pattern of 
\begin{quote}
\begin{center}
formal proof + soundness theorem $\mapsto$ realizing program
\end{center}
\end{quote}
which guarantees a realizing term from any proof. Our aim is to use oracle sequential algorithms to describe the operational behaviour of programs which realize the n.c.i. of theorems. To this end, rather than a full soundness proof, we focus on interpreting key patterns which often appear in proofs. As stated at the beginning, the idea is that a human being can produce simple algorithmic descriptions of mathematically simple parts of a proof, and then appeal to a number of formal constructions in order to deal with complicated parts of the proof, and such a construction, which deals with dependent choice, is the main result of the next section. Therefore the technique we illustrate in this paper is better described as
\begin{equation*}
\mbox{informal proof + constructions on algorithms $\mapsto$ realizing algorithm}
\end{equation*}
Nevertheless, this does not mean that it is not possible to formally extract algorithmic descriptions of realizing programs. Devices such as the Krivine machine \cite{Krivine(2009.0)} have already been developed for this purpose, where the reduction of extracted terms is animated through an abstract machine. More closely related to the work here is the variant of the functional interpretation developed by the author in \cite{Powell(2018.1)} which keeps track of case distinctions made by formally extracted realizing terms in a global state. In the next section, we pause and demonstrate how the reduction of type $2$ terms in System T can be captured by our sequential algorithms.

%%%%%%%%%%%%%%%%%%%%%%%%%%%%%%%%%%%%%%%%%%%%%%%%%
\subsection{Functionals in System T}
%%%%%%%%%%%%%%%%%%%%%%%%%%%%%%%%%%%%%%%%%%%%%%%%%
\label{sec-thmalg-syst} 

This section, which is inspired by a classic result which can be found in Troelstra \cite{Troelstra(1973.0)}, constitutes a short digression and can be skipped if desired. It contains a sketched proof of the following result:
\begin{proposition}
Suppose that $A:\equiv \forall u\exists x\forall y P(u,x,y)$ is provable in Peano arithmetic. Then there exists a computable and terminating approximation algorithm $\A$ that satisfies $P$, and which can be formally obtained from the proof of $A$.
\end{proposition}

As outlined in \cite[Proposition 10.9]{Kohlenbach(2008.0)}, one can first use G\"{o}del's functional interpretation to extract a closed term $t:(0\to 0)\to 0\to 0$ of System T which realizes the n.c.i. of $P$ i.e. $\forall u,\phi^{0\to 0} P(u,t\phi u,\phi(t\phi u))$. 

We define $\A$ of sort $\NN,\NN,\NN$ as follows: Let $R:=\{\cq,\cc\}\times T_\phi\times\NN$, where $\cs$ and $\cc$ are just symbols and $T_\phi$ denotes the set of all terms of System $T$ of type $0$ which contain a single free variable parameter, along the lines of \cite[Theorem 2.3.9]{Troelstra(1973.0)}, we consider terms of the form $\phi\num{n}$ to be additional `oracle' reducts, where $\num{n}$ denotes the numeral representation of $n\in\NN$. 

We now let query states be those of the form $\stac{\cq,a,x}{\Box}$ and end states those of the form $\stac{\cq,\num{k},x}{y}$, with $\rho(u):=(\cc,t\phi\num{u},0)$ and $\xi(z,a,x):=x$. Finally, the transition relation $\rhd$ is defined as:
\begin{equation*}
\begin{aligned}
\stac{\cc,a,x}{o}&\rhd\begin{cases}\stac{\cc,b,x}{o} & \mbox{if $a\succ b$}\\ \stac{\cq,a,n}{\Box} & \mbox{if the next reduct of $a$ is an oracle reduct $\phi\num{n}$}\\
\stac{\cq,a,n}{\Box} & \mbox{if $a=\num{n}$}\end{cases} \\
\stac{\cq,a,x}{y}&\rhd \stac{\cc,a[\num{y}/\phi\num{n}],x}{y} \ \ \  \mbox{where the next reduct of $a$ is $\phi\num{n}$}
\end{aligned}
\end{equation*}
Now, suppose that $f:\NN\to\NN$ is some function, and we extend our reduction strategy $\succ$ to $\succ_f$, so that it also acts on oracle reducts as $\phi\num{n}\succ_f \num{m}$ for $m=f(n)$. We show that if $a\succ_f b$ then for any state $s=\stac{\cc,a,x}{o}$ we would have $s\rhd_f^\ast \stac{\cc,b,x'}{o'}$. If $a\succ b$ via an non-oracle reduction this is obvious. Otherwise, we would have
\begin{equation*}
\stac{\cc,a[\phi \num{n}],x}{o}\rhd \stac{\cq,a[\phi\num{n}],n}{\Box}\orhd_f \stac{\cq,a[\phi\num{n}],n}{m}\rhd \stac{\cc,a[\num{m}/\phi\num{n}],n}{m}
\end{equation*}
with $m=f(n)$. By an argument similar to \cite[Theorem 2.3.9]{Troelstra(1973.0)}, any term in $T_\phi$ and so in particular $t\phi\num{u}$ would reduce via $\succ_f$ to a normal form of type $0$ i.e. a numeral $\num{n}$. Thus we have
\begin{equation*}
\rho_u=\stac{\cc,t\phi\num{u},0}{\Box}\rhd_f^\ast \stac{\cc,\num{n},x}{o}\rhd \stac{\cq,\num{n},n}{\Box}\orhd_f \stac{\cq,\num{n},n}{y}\in E
\end{equation*}
where $\num{n}$ is the normal form of $t$ and thus, considering $t$ as a functional $(\NN\to\NN)\to\NN$ in the full set-theoretic model of System T, we have $n=tfu$, $m=f(n)$ and thus $P(u,n,m)$. Moreover, the algorithm is computable since the reduction of a System T term relative to some strategy can be viewed as a computable operation, given some encoding of terms.

%%%%%%%%%%%%%%%%%%%%%%%%%%%%%%%%%%%%%%%%%%%%%%%%%
\subsection{Control flow graphs and mind-changes}
%%%%%%%%%%%%%%%%%%%%%%%%%%%%%%%%%%%%%%%%%%%%%%%%%
\label{sec-thmalg-graphs}

We now introduce a number of auxiliary concepts can be used to enhance our understanding of approximation algorithms. We begin with the notion of a control flow graph, which is valid for any of the algorithm types we have discussed in this section.

\begin{definition}[Control flow graph]
\label{def-graph}
Let $\A$ be an algorithm with states $S$, $\V$ a set and $\pi:S\to \V$ a function.
\begin{itemize}

\item For a normal (non-oracle) sequential algorithm, a directed graph $\G=(\V,\E)$ is is a control flow graph for $\A$ w.r.t. $\pi$ if whenever
\begin{equation*}
s_0,s_1,\ldots,s_n
\end{equation*}
is a run in $\A$ from some $s_0\in I$ then $(\pi(s_i),\pi(s_{i+1}))\in \E$ for all $i<n$.

\item For an oracle sequential algorithm, the labelled directed graph $\G=(\V,\E\cup \E')$ is a control flow graph for $\A$ w.r.t. $\pi$ if whenever
\begin{equation*}
s_0,s_1,\ldots,s_n
\end{equation*}
is a run in $\A$ on some oracle $f:X\to Y$ and some $s_0\in I_\A$ then $(\pi(s_i),\pi(s_{i+1}))\in A$ for all $i<n$ with $s_i\notin Q$, and $(\pi(s_i),\pi(s_{i+1}))\in A'$ for all $i<n$ with $s_i\in Q$.

\end{itemize}
\end{definition}

There is another useful concept which will apply to approximation algorithms.

\begin{definition}
Suppose that $\A$ is an approximation algorithm terminating on $\F$.
\begin{itemize}

\item If $x_0,\ldots,x_{m-1}$ is the query sequence of $\A$ on $u$ and $f$, we say that $\A$ requires $m-1$ mind-changes on $u$ and $f$.

\item We say that $\A$ requires at most $n$ mind-changes on $\F$ if it requires $\leq n$ mind changes for any $u$, $f\in F$.

\item We say that $\A$ is descending on $\F$ w.r.t. $\succ$ if whenever $s\rhd t$ then $\xi(s)\succeq \xi(t)$.

\end{itemize}
\end{definition}

%%%%%%%%%%%%%%%%%%%%%%%%%%%%%%%%%%%%%%%%%%%%%%%%%
\subsection{Multiple oracles and redundant parameters}
%%%%%%%%%%%%%%%%%%%%%%%%%%%%%%%%%%%%%%%%%%%%%%%%%
\label{sec-alg-mult}

Later it will be helpful to consider algorithms which computes functionals of type $U\times Y_1^X\times Y_2^X\to X$. While we could just encode this as a functional of type $U\times (Y_1\times Y_2)^X\to X$, for the purpose of efficiency and also intuition, it will be helpful to consider two separate oracles which can be queried independently, in which case we say that our algorithm has sort $U,X,[Y_1,Y_2]$.

We do not give a rigorous definition of the resulting extension of our oracle algorithms since this is fairly obvious. States would have the form $C\times U\times X\times Y_1\times Y_2$ while for each oracle we require separate subsets $Q_1,Q_2\subset S$ of query states with $Q_1\cap Q_2=\emptyset$ and $s\in Q_i\Rightarrow y_i[s]=\Box$. The second rule on oracle transitions would extend to at least one of the $y_i[s]$ being nonempty. Finally, there would be two transition relations $\orhd_i:Q_i\to S$ corresponding to each oracle.

In some cases we will consider theorems of the form $(\exists x)(\forall y) P(x,y)$, which are not dependent on some outer parameter $U$. Formally, these will be interpreted by OSAs of sort $1,X,Y$, where $1$ denotes some one element set (i.e. terminal object).

%%%%%%%%%%%%%%%%%%%%%%%%%%%%%%%%%%%%%%%%%%%%%%%%%
\subsection{Example: The least element principle}
%%%%%%%%%%%%%%%%%%%%%%%%%%%%%%%%%%%%%%%%%%%%%%%%%
\label{sec-thmalg-ex} 

We now give a concrete example which illustrates all of the above phenomena. Let $Q(x)$ be some decidable predicate on the set $X$, and assume this is equipped with a well-founded, decidable relation $>$. Then using classical logic plus induction over $>$, we can prove the following minimum principle:
\begin{equation*}
(\exists u\in X) Q(u)\to (\exists x\in X)(Q(x)\wedge (\forall y<x)\neg Q(y)).
\end{equation*}
Now, this is constructively equivalent to the following formula in prenex form:
\begin{equation*}
(\forall u\in X)(\exists x\in X)(\forall y\in X)(\underbrace{Q(u)\to Q(x)\wedge (y<x\to \neg Q(y))}_{P(u,x,y)}).
\end{equation*}
Let $P(u,x,y)$ be defined as indicated, and let the approximation algorithm $\A=(R,Q,E,\rho,\xi,\rhd)$ of sort $X,X,X$ be given as follows: 
\begin{itemize}

\item $R:=\{\cs,\ce\}\times X$ where $\cs$ and $\ce$ are symbols;

\item $Q$ consists of states of the form $\stac{\cs,x}{\Box}$;

\item $E$ contains states of the form $\stac{\ce,x}{y}$;

\item $\rho(u):=(\cs,u)$ and $\xi(c,x):=x$;

\end{itemize}
together with the single transition 
\begin{equation*}
\stac{\cs,x}{y}\rhd\begin{cases}\stac{\cs,y}{\Box} & \mbox{if $y<x\wedge Q(y)$}\\ \stac{\ce,x}{y} & \mbox{otherwise}.\end{cases}
\end{equation*}
\begin{proposition}
\label{prop-ex}
$\A$ satisfies $P(u,x,y)$ and terminates.
\end{proposition}

\begin{proof}
If $\A$ does not terminate on some $u$, we could generate an infinite sequence $u_>u_1>u_2>\ldots$ such that
\begin{equation*}
\stac{\cs,u}{\Box}\rhd^\ast_\omega \stac{\cs,u_i}{\Box}
\end{equation*}
for all $i$, contradicting well-foundedness of $>$. To show that the algorithm satisfies $P$, pick some $u\in X$ and assume that $P(u)$ holds (else the result is trivial). It is easy to show by induction on the run that whenever $\stac{\cs,u}{\Box}\rhd^\ast\stac{\cs,x}{\Box}$ then $P(x)$ holds. But since $\A$ terminates, we end up in some end state $\stac{\cs,x}{y}$ with $P(x)$ and $y<x\to \neg Q(y)$.
\end{proof}

The above program is an animation of the standard solution to the n.c.i. of the least element principle, which would be given by some lambda satisfying the following recursive equation:
\begin{equation*}
\Phi(f,x)=\mbox{ if $x>fx\wedge Q(fx)$ then $\Phi(f,fx)$ else $x$}.
\end{equation*}
Note that in general extracted terms are much more complicated than this, and as such we could use our framework to produce much richer algorithmic descriptions these terms, which abstract away inessential parts of the program and focus on the salient features.

We can associate to this approximation algorithm a control flow graph, which tracks the basic structure of the computation: We set $\Sigma:=\{\cs,\css,\ce\}$ and define $\pi:S\to \Sigma$ by
\begin{equation*}
\pi\stac{c,x}{y}:=c\mbox{ for $c\in \{\cs,\ce\}$, and }\pi\stac{\cs,x}{\Box}:=\css.
\end{equation*}
Then the directed graph $G=(\Omega,A\cup A')$, where $\cs A\css$, $\cs A\ce$ and $\css A'\cs$, i.e.
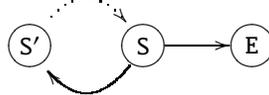
\begin{equation*}\xymatrix{*++[o][F]{\css}\ar@/^1.5pc/@{..>}[r] & *++[o][F]{\cs}\ar@/^1.5pc/^{}[l]\ar[r] & *++[o][F]{\ce}}\end{equation*}
is a control flow graph for $\A$. Moreover, $\A$ is descending w.r.t. $>$ and requires at most one mind change.

%%%%%%%%%%%%%%%%%%%%%%%%%%%%%%%%%%%%%%%%%%%%%%%%%
\subsection{The n.c.i in action: Extracting witnesses from non-constructive proofs}
%%%%%%%%%%%%%%%%%%%%%%%%%%%%%%%%%%%%%%%%%%%%%%%%%
\label{sec-thmalg-imp} 

In this section we discuss a very simple proof pattern, in which some non-constructive principle $A$ is used to prove an existential statement $\exists v Q(v)$. We indicate how $A$ can be replaced by its n.c.i. in the proof, and how an approximation algorithm for $A$ can be turned into a sequential algorithm for computing $v$. For simplicity we focus the case where $A:\equiv (\exists x\in X)(\forall y\in Y) P(x,y)$ for decidable $P$. Our aim is not a comprehensive study of program extraction using the n.c.i. (or closely related Dialectica interpretation). But our example covers a great number of non-constructive proofs one finds in everyday mathematics, particularly abstract algebra, in which one often a single instance of e.g. a maximal idea is combined with some elementary reasoning to prove the existence of something concrete.\medskip

\noindent\textbf{Theorem:} There exists some $v\in V$ such that $Q(v)$ holds.

\begin{proof}[\textbf{Proof pattern:}] Take a relevant instance $\exists x\in X\forall y\in Y P(x,y)$ of our non-constructive principle, and show that this combined with $\forall v \neg Q(v)$ leads to a contradiction.\end{proof}

\noindent Logically speaking, the body of the proof is the following implication:
\begin{equation*}
\exists x\forall y P(x,y)\wedge \forall v\neg Q(v)\to \bot
\end{equation*}
which is of course equivalent to
\begin{equation*}
(\ast) \ \ \ \exists x\forall y P(x,y)\to \exists v Q(v)
\end{equation*}
Thus if there exists some $x$ satisfying $\forall y P(x,y)$ then we can infer $\exists v Q(v)$. The problem, computationally speaking, is that we cannot necessarily \emph{construct} this $x$. However, we can typically find a pair of functions $f:X\to Y$ and $g:X\to V$ such that
\begin{equation*}
(\ast\ast) \ \ \ \forall x\left(P(x,f(x))\to Q(g(x))\right).
\end{equation*}
In fact, when $(\ast)$ is provable in e.g. Peano arithmetic, the functional interpretation guarantees that we are able to do so. Now supposing we are able to find a functional $\Phi:Y^X\to X$ satisfying
\begin{equation*}
P(\Phi f,f(\Phi f)).
\end{equation*}
Then it follows that $Q(v)$ holds for $v:=g(\Phi(f))$. We can give an algorithmic version if this as follows.

\begin{proposition}
\label{prop-imp}
Suppose that $\A$ of sort $1,X,Y$ is an approximation machine that terminates on $\F$ and satisfies $P(x,y)$, and that $f\in \F$ and $g:X\to V$ satisfy $(\ast\ast)$. Then $Q(v)$ holds for $v:=(g\circ \xi)(r)$ where
\begin{equation*}
\rho\rhd_f^\ast \stac{r}{y}\in E_\A.
\end{equation*}
\end{proposition}

\begin{proof}
Straightforward: If $\rho\rhd_f^\ast \stac{r}{y}$ then $P(x,y)$ holds for $x=\xi(r)$ and $y=f(\xi(r))$, and thus by $(\ast\ast)$ we have $Q(g(\xi(r)))$.
\end{proof}

As such, if $\stac{r_0}{\Box},\stac{r_1}{o_1},\ldots,\stac{r_{k-1}}{o_{k-1}},\stac{r_k}{y}\in E$ is a run in $\A$ on $f$ we can view the sequence
\begin{equation*}
(g\circ \xi)(r_0),\ldots,(g\circ \xi)(r_k)
\end{equation*}
as an algorithmic computation of a witness for $\exists v Q(v)$.

%%%%%%%%%%%%%%%%%%%%%%%%%%%%%%%%%%%%%%%%%%%%%%%%%
%%%%%%%%%%%%%%%%%%%%%%%%%%%%%%%%%%%%%%%%%%%%%%%%%
\section{An algorithmic interpretation of dependent choice}
%%%%%%%%%%%%%%%%%%%%%%%%%%%%%%%%%%%%%%%%%%%%%%%%%
%%%%%%%%%%%%%%%%%%%%%%%%%%%%%%%%%%%%%%%%%%%%%%%%%
\label{sec-dc}

In this section we give our main technical result: A formal translation on algorithms which acts as a computational interpretation of the rule of dependent choice. This result is interesting in its own right: While inspired by Spector's bar recursion, particularly the variant on finite sequences studied in \cite{OliPow(2012.2)}, our translation is not based on any predetermined formal system - everything works in our very general setting where algorithms are just transition relations which act on some state. In particular, we separate correctness and termination, first showing that our algorithm for dependent choice is correct whenever it terminates, and only then showing that it indeed terminates for a reasonable class of oracles. More importantly though, and unlike bar recursion, our interpretation gives a clear algorithmic meaning to dependent choice: Our construction results in a sequential algorithm which clearly builds an approximation to a choice sequence step-by-step, and can thus be used to better understand the computational content of proofs in analysis or abstract algebra which use make use of choice in some way. In particular, we are able to give a graph-theoretic formulation of our algorithm, which described how the transformation acts on call graphs. We give a concrete illustration of this in Section \ref{sec-tape}.

Let's first outline the problem we want to solve. We need some notation:
\begin{definition}
\label{def-seq}
For any set $X$, let $X^\ast$ denote the set of finite sequences over elements of $X$, where $\seq{}\in X^\ast$ is the empty sequence. The length of $a\in X^\ast$ is denoted by $|a|$, and the extension of $a$ with some $x\in X$ resp. $b\in X^\ast$ is denoted $a::x$ resp. $a::b$. Given an \emph{infinite} sequence $\alpha:X^\NN$ and $n\in\NN$ we define $\initSeg{\alpha}{n}\in X^\ast$ to be the initial segment of $\alpha$ of length $n$ i.e. $\initSeg{\alpha}{n}:=\seq{\alpha_0,\ldots,\alpha_{n-1}}$. 
\end{definition}
We aim to interpret the following rule of dependent choice restricted to $\Pi_1$ formulas: If for some predicate $P(u,x,y)$ we have
\begin{equation*}
(\forall u\in X^\ast)(\exists x\in X)(\forall y\in Y) P(u,x,y)
\end{equation*}
then we can infer
\begin{equation*}
(\exists \alpha:\NN\to X)(\forall n\in\NN,y\in Y)P(\initSeg{\alpha}{n},\alpha_n,y).
\end{equation*}
Note that even for $X=Y=\NN$, simple classical logic combined with this rule form of dependent choice is strong enough to give us arithmetical comprehension for $\Pi_1$-formulas, and also simple formulations of Zorn's lemma, single instances of which are used to prove many strong analytic principles. 

To interpret this rule in our framework of sequential algorithms, we mean the following: Given an approximation algorithm $\A$ of sort $X^\ast,X,Y$ satisfying $P(u,x,y)$, can we construct, in a uniform way, an approximation algorithm $\A_\omega$ of sort $1,X^\NN,[\NN,Y]$ which satisfies 
\begin{equation*}
P_\omega(\alpha,[n,y]):\equiv P(\initSeg{\alpha}{n},\alpha_n,y).
\end{equation*}
Of course, we could just define $\A_\omega$ so that it does not terminate on any input, in which case it trivially satisfies any formula! Therefore we also require the following property: Whenever $\A$ terminates then $\A_\omega$ terminates under some reasonable conditions. The latter will be established in Theorem \ref{thm-term}. But first to the definition of $\A_\omega$:
\begin{definition}
\label{def-baralg}
Suppose we are given some approximation algorithm $\A=(R,Q,E,\rho,\xi,\rhd)$ of sort $X^\ast,X,Y$. We define $\A_\omega=(R_\omega,Q_1,Q_2,E_\omega,\rho_\omega,\xi_\omega,\rhd_\omega)$ of sort $1,X^\NN,[\NN,Y]$ as follows:
\begin{itemize}

\item $R_\omega:=R^\ast\times X^\ast$. States will be written as $\stac{\sigma,a}{o_1,o_2}$ for $\sigma\in R^\ast, a\in X^\ast$, $o_1\in \NN_\Box$ and $o_2\in Y_\Box$;

\item $\rho_\omega(1):=([\rho([])],[])$ and thus the single initial state is $\stac{[\rho([])],[]}{\Box,\Box}$, which we denote by $\rho_\omega$;

\item $Q_1$ comprises states of the form $\stac{\sigma::r,[]}{\Box,\Box}$ with $\stac{r}{\Box}\in Q$;

\item $Q_2$ comprises states of the form $\stac{\sigma,[]}{n,\Box}$ with $n<|\sigma|$;

\item $E_\omega$ comprises states of the form $\stac{[],a}{n,y}$;

\item $\xi_\omega(\sigma,a):=\xi^\ast(\sigma)::a::\zero\in X^\NN$ where $\xi^\ast(\seq{\sigma_0,\ldots,\sigma_{k-1}}):=\seq{\xi(\sigma_0),\ldots,\xi(\sigma_{k-1})}$ and $\zero:=0,0,\ldots$.

\end{itemize}
Finally, the transitions $\rhd_\omega$ are given by three rule types:
\begin{enumerate}[(a)]

\item\label{item-ta} $\stac{\sigma,[]}{n,\Box}\rhd_\omega\stac{\sigma::\rho(\xi^\ast(\sigma)),[]}{\Box,\Box}$;

\item\label{item-tb} $\stac{\sigma::r,a}{n,y}\rhd_\omega \stac{\sigma,\xi(r)::a}{n,y}$ if $\stac{r}{y}\in E$;

\item\label{item-tc}
\begin{enumerate}[(i)]

\item\label{item-tci} $\stac{\sigma::r,[]}{\Box,\Box}\rhd_\omega \stac{\sigma::r',[]}{\Box,\Box}$ if $\stac{r}{\Box}\rhd \stac{r'}{\Box}$,

\item\label{item-tcii} $\stac{\sigma::r,a}{n,y}\rhd_\omega \stac{\sigma::r',[]}{\Box,\Box}$ if $\stac{r}{y}\rhd \stac{r'}{\Box}$,

\item\label{item-tciii} $\stac{\sigma::r,a}{n,y}\rhd_\omega \stac{\sigma::r',a}{n,y}$ if $\stac{r}{y}\rhd \stac{r'}{y}$.

\end{enumerate}
\end{enumerate}

\end{definition}

\begin{remark}
Note that $\A_\omega$ has multiple oracles, as outlined in Section \ref{sec-alg-mult}. In particular, we denote the two oracle transitions by $\orhd_{i}$ for $i=1,2$, where as usual we suppress the dependency on the particular oracle.
\end{remark}

\begin{lemma}
\label{lem-baraa}
$\A_\omega$ is an approximation algorithm.
\end{lemma}

\begin{proof}
We only need to establish condition (\ref{item-aaiii}) of Definition \ref{def-aa}. There are two transition rules to check: for (\ref{item-tb}) we have 
\begin{equation*}
\xi_\omega(\sigma::r,a)=\xi^\ast(\sigma)::\xi(r)::a=\xi_\omega(\sigma,\xi(r)::a)
\end{equation*}
while for (\ref{item-tciii}) we have
\begin{equation*}
\xi_\omega(\sigma::r,a)=\xi^\ast(\sigma)::\xi(r)::a=\xi^\ast(\sigma)::\xi(r')::a=\xi_\omega(\sigma::r',a)
\end{equation*}
using the fact that $\A$ is an approximation algorithm and hence $\xi(r)=\xi(r')$.
\end{proof}

\begin{lemma}
\label{lem-barcomp}
If $\A$ is computable then so is $\A_\omega$.
\end{lemma}

\begin{proof}
This is clear and we only sketch the proof: Membership of $Q_1,Q_2$ and $E_\omega$ only involves checking the lengths $|\sigma|$ and $|a|$ together with membership of $Q$. Similarly, establishing which transition rule applies is a matter of lengths together with membership of $E$, and $\rhd_\omega$ then only involves simple manipulations of the state, potentially using the computable functions $\rho$ and $\xi$, or the partial computable transition relation $\rhd$.
\end{proof}

We now prove that whenever $\A$ satisfies $P$, then $\A_\omega$ satisfies $P_\omega$. We do this by induction on the run, and we first need a few lemmas.

\begin{definition}
\label{def-level}
We say that a transition acting on state $\stac{\sigma,a}{o_1,o_2}$ has level $m$ if it either an oracle transition or an instance of rule (\ref{item-ta}), and $|\sigma|=m$, or it is an instance of rules (\ref{item-tb})-(\ref{item-tc}) and $|\sigma|=m+1$.
\end{definition}

\begin{lemma}
\label{lem-dc1}
Define $s_\sigma:=\stac{\sigma,[]}{\Box,\Box}$ and fix some oracles for $\A_\omega$. Suppose that $s_\sigma\rhd^\ast_\omega \stac{\sigma::\tau,a}{o_1,o_2}$ where all transitions in the run have level $\geq |\sigma|$. Then whenever $\xi^\ast(\sigma)=\xi^\ast(\sigma')$ we have $s_{\sigma'}\rhd^\ast_\omega \stac{\sigma'::\tau,a}{o_1,o_2}$.
\end{lemma}

\begin{proof}
Easy induction on the run. For the induction step, in the case of oracle transitions, these can only depend on $\xi_\omega(\sigma::\tau,a)$ and hence $\xi^\ast(\sigma)$, and similarly for rule (\ref{item-ta}). For rules (\ref{item-tb}) and (\ref{item-tc}), since the level of the transition is $\geq |\sigma|$ we must have $|\tau|\geq 1$ and hence the transitions are independent of $\sigma$.
\end{proof}

\begin{lemma}
\label{lem-dc2}
Fixing some oracles, suppose that the state $\stac{\sigma::r,a}{n,y}$ for $\stac{r}{y}\in E$ is reachable. Then $P(\xi^\ast(\sigma),\xi(r),y)$.
\end{lemma}

\begin{proof}
We fix some reachable $\stac{\sigma::r,a}{n,y}$ with $\stac{r}{y}\in E$, and first prove a simple yet bureaucratic auxiliary claim, namely that there is a partial function $f_\sigma:X\to Y$ such that if $\stac{\sigma::r',[]}{\Box,\Box}$ for $\stac{r'}{\Box}\in Q$ is some intermediate state on the run to $\stac{\sigma::r,a}{n,y}$, then $\stac{\sigma::r',[]}{\Box,\Box}\rhd_\omega^\ast \stac{\sigma::r',a'}{n',f_\sigma(\xi(r))}$ for some $a'$ and $n'$. 

To prove the claim, suppose that $x\in X$ is such that $\stac{\sigma::r',[]}{\Box,\Box}\in Q_1$ for some $r'$ with $\xi(r')=x$. Then  we must have $\stac{\sigma::r',[]}{\Box,\Box}\rhd^\ast_\omega \stac{\sigma::r',a'}{n',y'}$ for some $a',n',y'$, using transitions only of level $\geq |\sigma|+1$. To see this, note that either
\begin{equation*}
\stac{\sigma::r',[]}{\Box,\Box}\orhd_1\stac{\sigma::r',[]}{n',\Box}\orhd_2 \stac{\sigma::r',[]}{n',y'}
\end{equation*}
for $n'<|\sigma|+1$, or
\begin{equation*}
\stac{\sigma::r',[]}{\Box,\Box}\orhd_1\stac{\sigma::r',[]}{n',\Box}\rhd_\omega \stac{\sigma::r'::\rho(\xi^\ast(\sigma::r')),[]}{\Box,\Box}
\end{equation*}
and so at so at some point an instance of rule (\ref{item-tb}) must occur, and the first instance would result in a state of the form $\stac{\sigma::r',a'}{n',y'}$ having used only transitions of level $\geq |\sigma|+1$. 

Thus, we can set $f_\sigma(x):=y'$ for this $y'$, which is independent of our choice of $r'$ satisfying $\xi(r')=x$ by Lemma \ref{lem-dc1}, and undefined if there is no reachable state $\stac{\sigma::r',[]}{\Box,\Box}\in Q_1$ with $\xi(r')=x$. This proves the claim.

Now we utilise this function $f_\sigma$ to show that for any $u\in X^\ast$, we have
\begin{equation*}
\stac{\rho(u)}{\Box}\rhd^\ast \stac{r'}{y'}\mbox{ resp.} \stac{r}{\Box}
\end{equation*}
on oracle $f_\sigma$, if and only if
\begin{equation*}
\stac{\sigma::\rho(u),[]}{\Box,\Box}\rhd^\ast_\omega \stac{\sigma::r',a'}{n',y'}\mbox{ resp.} \stac{\sigma::r',[]}{\Box,\Box}.
\end{equation*}
for some $a',n'$. But this is a simple induction on runs which follows by inspecting rule type (\ref{item-tc}). The only interesting case is an oracle transition $\stac{r'}{\Box}\orhd \stac{r'}{y'}$ for $\stac{r'}{\Box}\in Q$, but this is exactly what is covered by the claim.

We can now complete the proof: If $\stac{\sigma::r,a}{n,y}$ is reachable we must, more specifically, have
\begin{equation*}
\rho_\omega(1)\rhd^\ast_\omega \stac{\sigma::\rho(\xi^\ast(\sigma)),[]}{\Box,\Box}\rhd_\omega^\ast \stac{\sigma::r,a}{n,y}
\end{equation*}
and thus $\stac{\rho(\xi^\ast(\sigma))}{\Box}\rhd^\ast \stac{r}{y}\in E$ in $\A$ on $f_\sigma$. But since we assume that $\A$ satisfies $P$ we have $P(\xi^\ast(\sigma),\xi(r),y)$.
\end{proof}

\begin{lemma}
\label{lem-dc3}
Suppose that the state $\stac{\sigma,a}{n,y}$ is reachable. Then we have
\begin{equation*}
|\sigma|\leq n\to P(\initSeg{\alpha}{n},\alpha_n,y)\mbox{ \ \ for \ \ }\alpha:=\xi_\omega(\sigma::a).
\end{equation*}

\end{lemma}
\begin{proof}
Induction on the length of the run, fixing some oracles. The base case is trivial, so let's assume that $\rho_\omega(1),\ldots,s,\stac{\sigma,a}{n,y}$ is a run in $\A$. Then the last transition can only be an oracle transition on the second oracle, or an instance of rule (\ref{item-tb}) or (\ref{item-tciii}). 

In the former case, we would have $s=\stac{\sigma,[]}{n,\Box}\orhd_2 \stac{\sigma,[]}{n,y}$ with $n<|\sigma|$, and the result trivially holds. For rule (\ref{item-tciii}), we would have $s=\stac{\sigma'::r,a}{n,y}\orhd_2 \stac{\sigma'::r',a}{n,y}$ with $\stac{r}{y}\rhd\stac{r'}{y}$ and hence $\xi(r)=\xi(r')$, so the result holds by induction. Finally, for rule (\ref{item-tb}), we have $s=\stac{\sigma::r,a}{n,y}$ with $\stac{r}{y}\in E$ and $|\sigma|+1\leq n\to P(\initSeg{\alpha}{n},\alpha_n,y)$ for $\alpha:=\xi_\omega(\sigma::r::a)=\xi_\omega(\sigma::\xi(r)::a)$. But by Lemma \ref{lem-dc2}  we must also have $P(\xi^\ast(\sigma),\xi(r),y)$, and hence $|\sigma|=n\to P(\initSeg{\alpha}{n},\alpha_n,y)$, and we're done.
\end{proof}

\begin{theorem}
\label{thm-dc}
Suppose that $\A$ satisfies $P$. Then $\A_\omega$ satisfies $P_\omega$. 
\end{theorem}

\begin{proof}
Fix some oracles for $\A_\omega$ and suppose that $\rho_{\omega}(1)\rhd_\omega^\ast\stac{[],a}{n,y}\in E$. Then by Lemma \ref{lem-dc3} we have $0\leq n\to P(\initSeg{\alpha}{n},\alpha_n,y)$ for $\alpha=\xi_\omega([],a)=a::\zero$ which is just $P_\omega(\xi_\omega([],a),[n,y])$.
\end{proof}

We now turn to under what conditions our algorithm $\A_\omega$ actually terminates, assuming $\A$ that does.

\begin{lemma}
\label{lem-term1}
Suppose that $\A$ terminates. Define the predicate $B(\sigma)$ on $R^\ast$ as follows:
\begin{quote}$B(\sigma)$ holds only if $\stac{\sigma_i}{\Box}\in Q$ for all $i<|\sigma|$ and it is not the case that $\stac{\sigma,[]}{\Box,\Box}\rhd_\omega^\ast \stac{\sigma,a}{n,y}$ for some $a,n,y$.\end{quote}
Then $B(\sigma)$ implies there exists some $r\in R$ such that $B(\sigma::r)$. 
\end{lemma}

\begin{proof}
This is a simple adaptation of the argument already seen in Lemma \ref{lem-dc2}. If $B(\sigma)$ holds then we must have
\begin{equation*}
\stac{\sigma,[]}{\Box,\Box}\orhd_1\stac{\sigma,[]}{n,\Box}\orhd_2 \stac{\sigma::\rho(\xi^\ast(\sigma)),[]}{\Box,\Box}
\end{equation*}
with $n\geq |\sigma|$. Now, if for all $r$ with $\stac{r}{\Box}\in Q$ we have $\neg B(\sigma::r)$, this would imply that $\stac{\sigma::r,[]}{\Box,\Box}\rhd_\omega^\ast\stac{\sigma::r,a}{n,y}$ for some $n,a,y$ dependent only on $\xi^\ast(\sigma::r)$. Therefore there is an induced run $\stac{\rho(\xi^\ast(\sigma))}{\Box}\rhd^\ast \stac{r}{y}\in E$ in $\A$, which gives rise to a corresponding run $\stac{\sigma::\rho(\xi^\ast(\sigma)),[]}{\Box,\Box}\rhd_\omega^\ast\stac{\sigma::r,a}{n,y}$ in $\A_\omega$. But since $\stac{\sigma::r,a}{n,y}\rhd_\omega \stac{\sigma,\xi(r)::a}{n,y}$, this would contradict $B(\sigma)$.
\end{proof}

\begin{lemma}
\label{lem-term2}
Suppose that $\A$ terminates but that $\A_\omega$ does not terminate on oracles $\phi_1:X^\NN\to\NN$ and $\phi_2:X^\NN\to Y$. Then there exists some $\gamma:\NN\to X$ such that
\begin{equation*}
(\forall N\in\NN)(N\leq \phi_1(\initSeg{\gamma}{N}::\zero))
\end{equation*}
\end{lemma}

\begin{proof}
If $\A_\omega$ does not terminate then there is some $\stac{r}{\Box}\in Q$ such that $B([r])$. Applying Lemma \ref{lem-term1} together with dependent choice on the meta-level, we obtain a sequence $F:\NN\to R$ such that for all $N\in \NN$, $\stac{F_N}{\Box}\in Q$ but it is not the case that $\stac{\initSeg{F}{N},[]}{\Box,\Box}\rhd_\omega^\ast \stac{\initSeg{F}{N},a}{n,y}$. 

In particular, we must have $N\leq \phi_1(\xi_\omega(\initSeg{F}{N},[]))$, else $\stac{\initSeg{F}{N},[]}{\Box,\Box}\orhd_1 \stac{\initSeg{F}{N},[]}{n,\Box}\in Q_2$. Define $\gamma_n:=\xi(F_n)\in X$. Then $\xi_\omega(\initSeg{F}{N},[])=\initSeg{\gamma}{N}::\zero$, and the result follows.
\end{proof}
Our main theorem on termination now follows directly from Lemma \ref{lem-term2}:
\begin{theorem}
\label{thm-term}
Suppose that $\A$ terminates. Then $\A_\omega$ terminates on $\F_1\subset X^\NN\to\NN$ and the full function space $X^\NN\to Y$, where
\begin{equation*}
\F_1:=\{\phi:X^\NN\to\NN\; | \; (\forall \gamma)(\exists N)(N<\phi_1(\initSeg{\gamma}{N}::\zero))\}
\end{equation*}
\end{theorem}

While $\F_1$ may not seem at first glace to be a particularly natural restriction, note that any \emph{continuous} $\phi_1$ is automatically a member of $\F_1$. To see this: if some continuous $\phi_1$ takes an argument $\gamma$, then $\phi_1(\gamma)$ depends only on some finite initial segment $\initSeg{\gamma}{n}$ of $\gamma$. Taking $N:=\max(\phi_1(\gamma)+1,n)$, we have
\begin{equation*}
\phi_1(\initSeg{\gamma}{N}::\zero)=\phi_1(\gamma)<\phi_1(\gamma)+1=N.
\end{equation*}

%%%%%%%%%%%%%%%%%%%%%%%%%%%%%%%%%%%%%%%%%%%%%%%%%
\subsection{Understanding how choice sequences are built}
%%%%%%%%%%%%%%%%%%%%%%%%%%%%%%%%%%%%%%%%%%%%%%%%%
\label{sec-dc-understand}

While our main construction of $\A_\omega$ may seem quite involved, the idea is that given an intelligent and meaningful input algorithm $\A$, we are able to clearly see how approximations to choice sequences are build by $\A_\omega$. In particular, the concepts introduced in Section \ref{sec-thmalg-graphs} can both be lifted to this setting.

\begin{theorem}
\label{thm-graph}
Suppose that $G=(\Omega,A\cup A')$ is a control flow graph for $\A$ w.r.t. $\pi:S\to \Omega$, and let $\Omega_I:=\pi(\rho(X^\ast))$ and $\Omega_E:=\pi(E)$. We define
\begin{equation*}
\Sigma:=\{\star\}+((\bt{\NN}+\tp{\NN})\times\Omega)
\end{equation*}
where $\bt{\NN}$ and $\tp{\NN}$ are two copies of $\NN$, and the directed graph $G_\omega:=(\Sigma,B\cup B')$ by
\begin{enumerate}

\item\label{item-ga} if $pAq$ then $(\bt{n},p)B(\bt{n},q)$;

\item\label{item-gb} if $pA'q$ then:
\begin{itemize}

\item $(\bt{n},p)B'(\tp{n},p)$;

\item $(\tp{n},p)B'(\bt{n},q)$;

\item $(\tp{n},p)B(\bt{n+1},u)$ for all $u\in \Omega_I$

\item $(\bt{n+1},u)B(\bt{n},q)$ for $u\in \Omega_E$;

\end{itemize}

\item $(\bt{0},u)B\star$ for $u\in \Omega_E$.

\end{enumerate}
Finally, let $\pi_\omega$ be given by
\begin{equation*}
\begin{aligned}
\pi_\omega\stac{\sigma::r,a}{o_1,o_2}&:=\begin{cases}(\tp{|\sigma|},\pi\stac{r}{o_2}) & \mbox{if $o_1\in\NN,o_2=\Box$}\\
(\bt{|\sigma|},\pi\stac{r}{o_2}) & \mbox{otherwise}\end{cases}\\ 
\pi_\omega\stac{[],a}{o_1,o_2}&=\star
\end{aligned}
\end{equation*}
Then $G_\omega$ is a control flow graph for $\A_\omega$ w.r.t $\pi_\omega$.
\end{theorem}

\begin{proof}
Suppose that $\rho_\omega(1),\ldots,s,t$ is a run in $\A_\omega$. If $s\rhd_\omega t$ by rule type (\ref{item-tc}) then $\pi_\omega(s)=(\bt{n},p)$ and $\pi_\omega(t)=(\bt{n},q)$ for $pAq$ and some $\bt{n}\in\bt{\NN}$, and thus $(\bt{n},p)B(\bt{n},q)$.

If $s\rhd_\omega t$ by rule type (\ref{item-tb}), there are two cases. Either $\pi_\omega(s)=(\bt{n+1},u)$ for $u\in \Omega_E$, in which case $\pi_\omega(t)=(\bt{n},q)$ where $pA'q$ for some $p$, and thus $(\bt{n+1},u)B(\bt{n},q)$. Or $\pi_\omega(s)=(\bt{0},u)$ for $u\in\Omega_E$, in which case $\pi(t)=\star$ and so $(\bt{0},u)B\star$.

If $s\orhd_1 t$, then $\pi_\omega(s)=(\bt{n},p)$ where $pA'q$ for some $q$ and $\pi_\omega(t)=(\tp{n},p)$, and we have $(\bt{n},p)B'(\tp{n},p)$.

If $s\rhd_\omega t$ by rule type (\ref{item-ta}), then $s$ must be the result of an oracle query from some $s'\in Q_1$ and thus $\pi_\omega(s)=(\tp{n},p)$ where $pA'q$ for some $q$. But $\pi_\omega(t)=(\bt{n+1},u)$ for some $u\in\Omega_I$, and $(\tp{n},p)B(\bt{n+1},u)$.

Finally, if $s\orhd_2 t$ then $s$ must be the result of an oracle query from some $s'\in Q_1$ and thus $\pi_\omega(s)=(\tp{n},p)$. But $\pi_\omega(t)=(\bt{n},q)$ where $pA'q$ and thus $(\tp{n},p)B'(\bt{n},q)$.
\end{proof}

We can also characterise how mind changes work for choice sequences.

\begin{theorem}
\label{thm-desc}
Suppose that $\A$ is descending w.r.t. $(W,\succ)$. Define the relation $\succ_\omega$ on $X^\ast$ by
\begin{itemize}

\item $\seq{x_0,\ldots,x_{k-1}}\succ_\omega \seq{x_0,\ldots,x_k,x}$ for $x\in X$;

\item $\seq{x_0,\ldots,x_{i-1},x,\ldots,x_{k-1}}\succ_\omega \seq{x_0,\ldots,x_{i-1},y}$ for $x\succ y$.

\end{itemize}
Extend this to a relation on $X^\NN$ by setting $\alpha\succ_\omega\beta$ iff $\alpha=a::\zero$ and $\beta=b::\zero$ with $a\succ_\omega b$. Then $\A_\omega$ is descending w.r.t. $\succ_\omega$.
\end{theorem}

\begin{proof}
This follows easily by inspecting each of the rule types.
\end{proof}

Producing explicit upper bounds on mind changes is a little more difficult: even if $\A$ makes at most one mind change, $\A_\omega$ could make unbounded many. However, if we restrict the first oracle so that $\A_\omega$ only considers registers $R^\ast$ of length at most $N$, we obtain some simple bounds.

For example, suppose that $\A$ terminates and requires at most $h(|u|)$ mind changes on input $u\in X^\ast$, and that $\F_1^N:=\{\phi_1:X^\NN\to\NN\; \ \; (\forall \gamma)(N<\phi_1(\initSeg{\gamma}{N}::\zero))0\}$. Then $\A_\omega$ terminates and requires at most
\begin{equation*}
\sum_{j=0}^N\prod_{i=0}^j h(i)
\end{equation*}
calls to the first oracle, and at most
\begin{equation*}
\prod_{i=0}^N h(i)
\end{equation*}
calls to the second, when run on $\phi_1\in \F_1^N$. We leave these as claims without proofs, since we do not want to go into the messy details: We simply wanted to give a sense of what could be obtained with further analysis.

%%%%%%%%%%%%%%%%%%%%%%%%%%%%%%%%%%%%%%%%%%%%%%%%%
%%%%%%%%%%%%%%%%%%%%%%%%%%%%%%%%%%%%%%%%%%%%%%%%%
\section{Case study: The infinite tape}
%%%%%%%%%%%%%%%%%%%%%%%%%%%%%%%%%%%%%%%%%%%%%%%%%
%%%%%%%%%%%%%%%%%%%%%%%%%%%%%%%%%%%%%%%%%%%%%%%%%
\label{sec-tape}

We conclude the article with a concrete example, which captures all of the ideas presented in the previous section, and moreover illustrates the approach to program extraction in general we outlined in Section \ref{sec-thmalg}. Our example is the following:
\begin{theorem}
\label{thm-tape}
Given some infinite sequence of booleans $b\in\{0,1\}^\NN$, for any $N\in\NN$ there is a constant subsequence $b_{v_0},\ldots,b_{v_{N-1}}$ of length $N$.
\end{theorem}
This is the so-called 'infinite tape' example, and was chosen due it its simplicity, and the fact that its constructive content has been studied in several places. Our formulation is most closely related to that of Seisenberger \cite[Section 3]{Seisenberger(2008.0)}: We will analyse a variant of the proof given there, which is adapted so that dependent choice is used in its rule form. Note that dependent choice is of course not needed to prove Theorem \ref{thm-tape}, which can be formalised in arithmetic without recourse to set theoretic axioms. Nevertheless, similarly to \cite{Seisenberger(2008.0)}, our analysis here uses $\A_\omega$ in a weak way, finitary way, and thus produces a simple program which could be formalised using primitive recursion.
\begin{proof}[Proof of Theorem \ref{thm-tape}]
By dependent choice, construct an increasing function $\alpha:\NN\to\NN$ as follows: Suppose $\alpha_0,\ldots,\alpha_{n-1}$ has already been chosen. Either there exists some $i>\alpha_{n-1}$ with $b_i=0$, in which case define $\alpha_n:=i$, or $b_i=1$ for all $i>\alpha_{n-1}$, in which case define $\alpha_n:=\alpha_{n-1}+1$. Then $\alpha$ is a monotone function such that, for any $n\in\NN$,
\begin{equation*}
b_{\alpha_n}=1\to \forall k\geq \alpha_{n}(b_k=1).
\end{equation*}
Now either $b_{\alpha_0}=\ldots=b_{\alpha_{N-1}}=0$, or there is some $n<N$ with $b_{\alpha_n}=1$, in which case $b_{\alpha_n}=b_{\alpha_n+1}=\ldots=b_{\alpha_n+N-1}=1$. Either way, we have found a constant subsequence of length $N$.
\end{proof}

The proof above is characteristic of many proofs which use dependent choice, and has exactly the kind of structure which lends itself to our approach: The instance of choice is essentially the only difficult non-constructive feature of the proof, and following the pattern presented in Section \ref{sec-thmalg-imp} we can interpret the surrounding reasoning fairly straightforwardly, and without recourse to any complicated machinery. Then, utilising our construction of the previous section, we obtain an elegant and transparent algorithm which finds the required constant subsequence by building an approximation to the choice sequence used in the classical proof.

%%%%%%%%%%%%%%%%%%%%%%%%%%%%%%%%%%%%%%%%%%%%%%%%%
\subsection{The structure of the proof}
%%%%%%%%%%%%%%%%%%%%%%%%%%%%%%%%%%%%%%%%%%%%%%%%%
\label{sec-tape-structure}

We begin by showing that the proof of Theorem \ref{thm-tape} fits the basic pattern outlined in Section \ref{sec-thmalg-imp}. First, we fix the input data for the theorem, namely a sequence $b\in\{0,1\}^\NN$ and a number $N\in\NN$ as a global parameters. We then define the decidable predicate $P$ on $\NN^\ast\times\NN\times\NN$ as
\begin{equation*}
P([u_0,\ldots,u_{n-1}],x,y):\equiv u_{n-1}<x\wedge ((b_x=1\wedge x\leq y)\to b_y=1).
\end{equation*}
where for $P([],x,y)$ we omit the conjunct $u_{n-1}<j$, and the decidable predicate $Q$ on $\NN^\ast$ as
\begin{equation*}
Q(v):=|v|=N\wedge \forall i<N(v_i<v_{i+1}\wedge b_{v_i}=b_{v_{i+1}})
\end{equation*}
Note that if $v$ satisfies $Q(v)$ then it gives us our constant subsequence of length $N$, and hence Theorem \ref{thm-tape} is just the statement $(\exists v\in\NN^\ast) Q(v)$. The formal structure of our proof can now be given precisely as the implication
\begin{equation*}
(\ast) \ \ \ (\exists\alpha:\NN\to\NN)(\forall n,y\in\NN)P(\initSeg{\alpha}{n},\alpha_n,y)\to (\exists v\in\NN^\ast) Q(v).
\end{equation*}
To see this, note that by the premise of $(\ast)$ there exists a sequence $\alpha$ satisfying
\begin{equation*}
\alpha_{n-1}<\alpha_n\wedge ((b_{\alpha_n}=1\wedge\alpha_n\leq y)\to b_y=1)
\end{equation*}
for all $n,y$. Either $b_{\alpha_0}=\ldots=b_{\alpha_{N-1}}=0$, in which case since $\alpha_0<\alpha_1<\ldots<\alpha_{N-1}$ we can set $v:=[\alpha_0,\ldots,\alpha_{N-1}]$, or there is some $m<N$ such that $b_{\alpha_m}=1$, in which case $b_{\alpha_m}=b_{\alpha_{m}+1}=\ldots=b_{\alpha_{m}+N-1}=1$ and we can set $v:=[\alpha_m,\alpha_m+1,\ldots,\alpha_m+N-1]$.

%%%%%%%%%%%%%%%%%%%%%%%%%%%%%%%%%%%%%%%%%%%%%%%%%
\subsection{Interpreting the implication}
%%%%%%%%%%%%%%%%%%%%%%%%%%%%%%%%%%%%%%%%%%%%%%%%%
\label{sec-tape-imp}

Referring back to the discussion in Section \ref{sec-thmalg-imp}, our next step is to find a pair of functions $f:\NN^\NN\to \NN\times\NN$ and $g:\NN^\NN\to \NN^\ast$ such that
\begin{equation*}
(\ast\ast) \ \ \ (\forall\alpha)(P(\initSeg{\alpha}{f_1(\alpha)},\alpha_{f_1(\alpha)},f_2(\alpha))\to Q(g(\alpha)),
\end{equation*}
where $f_i$ denotes the $i$th projection of $f$. The function $f$ characterises `how much' of the choice sequence $\alpha$ we actually need, while $g$ tells us how to convert this information into a witness for the theorem. Giving these functions is rather simple, since they can essentially be read of from the discussion in Section \ref{sec-tape-structure}. In short, we only need our $\alpha$ to satisfy $P(\initSeg{\alpha}{n},\alpha_n,y)$ up to $n:=N-1$ and $y:=\alpha_m+N-1$ where $m<N$ is such that $b_{\alpha_m}=1$. However, we need to be a little careful here: in order to ensure that our choice sequence is valid not just for these points but for everything up to these points we should incorporate a bounded search into our functions as follows:
\begin{lemma}
\label{lem-param}
Let $f$ and $g$ be defined as follows:
\begin{equation*}
\begin{aligned}
f_1(\alpha)&:=\mbox{least $n<N$ s.t. $\neg(\alpha_{n-1}<\alpha_n)\vee b_{\alpha_n}=1$, else $N-1$}\\
f_2(\alpha)&:=\mbox{if $b_{\alpha_n}=1$ then least $\alpha_n\leq y\leq\alpha_n+N-1$ s.t. $b_y=0$, else $\alpha_n+N-1$}\\
g(\alpha)&:=[\alpha_n,\alpha_n+1,\ldots,\alpha_n+N-1]\mbox{ if $b_{\alpha_n}=1$, else }[\alpha_0,\ldots,\alpha_{N-1}]
\end{aligned}
\end{equation*}
where $n:=f_1(\alpha)$ in the definitions of $f_2(\alpha)$ and $g(\alpha)$. Then $f$ and $g$ satisfy $(\ast\ast)$.
\end{lemma}

\begin{proof}
This is just a simple reformulation of the argument in Section \ref{sec-tape-structure}. Set $n:=f_1(\alpha)$. There are two cases. Suppose that $b_{\alpha_n}=1$, but there is some $\alpha_n\leq y\leq\alpha_n+N-1$ such that $b_y=0$. If this were the case, then $f_2(\alpha)$ would return this $y$, but $P(\initSeg{\alpha}{n},\alpha_n,y)$ would then be false, a contradiction. Therefore it must be the case that $b_y=1$ for all $\alpha_n\leq y\leq \alpha_n+N-1$, in which case $Q(v)$ holds.

On the other hand, suppose that $b_{\alpha_n}=0$. If $\neg(\alpha_{n-1}<\alpha_n)$ then this would contradict $P(\initSeg{\alpha}{n},\alpha_n,f_2(\alpha))$, and thus it must be the case that for all $n<N$ we have $(\alpha_{n-1}<\alpha_n)\wedge b_{\alpha_n}=0$, and thus $Q(v)$ holds.
\end{proof}

The functions defined in Lemma \ref{lem-param}, though not entirely trivial, are little more than bounded searches which check very basic properties of $\alpha$, and so in particular their algorithmic behaviour is clear. Thus, referring to Proposition \ref{prop-imp}, in order to obtain a full procedure which computes a witness $v$, it remains to find an approximation machine which satisfies $P(\initSeg{\alpha}{n},\alpha_n,y)$.

%%%%%%%%%%%%%%%%%%%%%%%%%%%%%%%%%%%%%%%%%%%%%%%%%
\subsection{Interpreting dependent choice}
%%%%%%%%%%%%%%%%%%%%%%%%%%%%%%%%%%%%%%%%%%%%%%%%%
\label{sec-tape-imp}

We now appeal to the machinery of Section \ref{sec-dc} to construct a machine which terminates and satisfies $P(\initSeg{\alpha}{n})$. For the latter property, it is enough by Theorem \ref{thm-dc} to construct a machine $\A$ which satisfies $P(u,x,y)$, and provided $\A$ terminates, termination of $\A_\omega$ follows from Theorem \ref{thm-term} and the fact that $f_1$ as defined in Lemma \ref{lem-param} is clearly continuous (in fact, it only ever looks at the first $N$ elements of its input $\alpha$).

So our challenge is to find a terminating machine $\A$ of sort $\NN^\ast,\NN,\NN$ which satisfies $P(u,x,y)$. But this is, again, rather straightforward. Intuitively: Given some input $u$ we first try $x:=u_{n-1}+1$. Given some $y$ returned by the oracle, $P(u,x,y)$ then holds unless $b_{x}=1\wedge x\leq y\wedge b_y=0$. But then taking $x:=y$ we trivially have $P(u,x,y')$ for any subsequent oracle answer $y'$, since $x>u_{n-1}$ and $b_x=0$ (for the simpler case $|u|=0$ we just take $x:=0$ for our first attempt). We now want to capture this idea as an approximation algorithm, which we can do by adapting that of Section \ref{sec-thmalg-ex}: Let $\A$ be given by:
\begin{itemize}

\item $R:=\{\cs,\ce_1,\ce_2\}\times \NN$ for symbols $\cs$ and $\ce$;

\item $Q$ consists of states of the form  $\stac{\cs/\ce_1,x}{\Box}$;

\item $E$ consists of states of the form $\stac{\ce_i,x}{y}$;

\item $\rho(u):=(\cs,u_{n-1}+1)$ with $\rho([]):=(\cs,0)$ and $\xi(c,x):=x$,

\end{itemize}
together with the transition
\begin{equation*}
\stac{\cs,x}{y}\rhd\begin{cases}\stac{\ce_1,y}{\Box} & \mbox{if $b_x=1\wedge x\leq y\wedge b_y=0$}\\
\stac{\ce_2,x}{y} & \mbox{otherwise}\end{cases}
\end{equation*}
\begin{lemma}
\label{lem-A}
$\A$ satisfies $P(u,x,y)$ and terminates.
\end{lemma}

\begin{proof}
We have $\stac{\cs,u_{n-1}+1}{\Box}\orhd\stac{\cs,u_{n-1}+1}{y}$ and there are two possibilities. Either $b_{u_{n-1}+1}=1\wedge u_{n-1}+1\leq y\to  b_{y}=1$ and $\stac{\cs,u_{n-1}+1}{y}\rhd \stac{\ce_2,u_{n-1}+1}{y}\in E$, and we have $P(u,u_{n-1}+1,y)$ by definition, or $b_{u_{n-1}+1}=1\wedge u_{n-1}+1\leq y\wedge  b_{y}=0$ and $\stac{\cs,u_{n-1}+1}{y}\rhd \stac{\ce_1,y}{\Box}\orhd \stac{\ce_1,y}{y'}$ and $P(u,y,y')$ holds since $u_{n-1}<y$ and $b_y=0$. Similar but easier for $u=[]$.
\end{proof}

We can now, in a single step, construct a machine which constructs our choice sequence.
\begin{theorem}
\label{thm-Aw}
Let $\A_\omega$ of sort $1,\NN^\NN,[\NN,\NN]$ be defined by
\begin{itemize}

\item $R_\omega:=(\{\cs,\ce_1,\ce_2\}\times\NN)^\ast\times\NN^\ast$;

\item $\rho_1:=\stac{[(\cs,0)],[]}{\Box,\Box}$;

\item $Q_1$ comprises states of the form $\stac{\sigma::(\cs/\ce_1,x),[]}{\Box,\Box}$;

\item $Q_2$ comprises states of the form $\stac{\sigma,[]}{n,\Box}$ with $n<|\sigma|$;

\item $E$ comprises states of the form $\stac{[],a}{n,y}$;

\item $\xi_\omega(\sigma,a):=\sigma_2::a::\zero$ where $\sigma_2\in\NN^\ast$ is the second projection of $\sigma$

\end{itemize}
with transitions given by
\begin{enumerate}[(a)]

\item $\stac{\sigma::(c,x),[]}{n,\Box}\rhd_\omega \stac{\sigma::(c,x)::(\cs,x+1),[]}{\Box,\Box}$;

\item $\stac{\sigma::(\ce_i,x),a}{n,y}\rhd_\omega \stac{\sigma,x::a}{n,y}$;

\item $\stac{\sigma::(\cs,x),a}{n,y}\rhd_\omega \begin{cases}\stac{\sigma::(\ce_1,y),[]}{\Box,\Box} & \mbox{if $b_x=1\wedge x\leq y\wedge b_y=0$}\\
\stac{\sigma::(\ce_2,x),a}{n,y} & \mbox{otherwise}\end{cases}$

\end{enumerate}
Then $\A_\omega$ satisfies $P(\initSeg{\alpha}{n},\alpha_n,y)$.
\end{theorem}

\begin{proof}
By Lemma \ref{lem-A} combined with Theorem \ref{thm-dc}.
\end{proof}

We are now able to give a clear, algorithmic interpretation of our proof of Theorem \ref{thm-tape}.

\begin{theorem}
\label{thm-tapealg}
Suppose that $\A_\omega$ is run on $f_1$ and $f_2$ as defined in Lemma \ref{lem-param}. Then 
\begin{equation*}
\stac{[(\cs,0)],[]}{\Box,\Box}\rhd_\omega^\ast \stac{[],a}{n,y}
\end{equation*}
where for $\alpha:=a::\zero$ we have that $b_{\alpha_n}=1$ implies $b_{\alpha_n}=b_{\alpha_n+1}=\ldots=b_{\alpha_n+N-1}$ and $b_{\alpha_n}=0$ implies $\alpha_0<\ldots,\alpha_{N-1}$ and $b_{\alpha_0}=\ldots=b_{\alpha_{N-1}}$.
\end{theorem}

\begin{proof}
Termination of $\A_\omega$ on $f_1$ follows from continuity of $f_1$ combined with Theorem \ref{thm-term}. The rest follows from Proposition \ref{prop-imp} and Lemma \ref{lem-param}. 
\end{proof}

Theorem \ref{thm-tapealg} was obtained using a combination of intuition (in working out the functions $f_1,f_2$ and $g$ together with the algorithm $\A$) and theory (the application of Theorems \ref{thm-dc} and \ref{thm-term}). The resulting procedure for computing a constant subsequence of $b$ has a clear description as a sequential algorithm which builds the underlying choice sequence.

We now discuss some of the additional concepts we introduced in Section \ref{sec-thmalg-graphs}. Similarly to Section \ref{sec-thmalg-ex}, we can construct a simple control flow graph for $\A$: Let $\Sigma:=\{\cs,\css,\ce,\ces\}$ and define $\sigma\stac{\cs,x}{\Box}:=\css$, $\sigma\stac{\cs,x}{y}:=\cs$, $\sigma\stac{\ce_1,x}{\Box}:=\ces$ and $\sigma\stac{\ce_i,x}{y}:=\ce$. Then the directed graph $G=(\Omega,A\cup A')$ with $\cs A\ce$, $\cs A \ces$, $\cs'A'\cs$ and $\ce'A'\ce$ i.e.
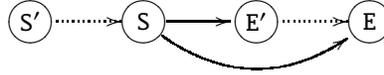
\begin{equation*}\xymatrix{*++[o][F]{\css}\ar@{..>}[r] & *++[o][F]{\cs}\ar@/_1.5pc/^{}[rr]\ar[r] & *++[o][F]{\ce'}\ar@{..>}[r] & *++[o][F]{\ce}}\end{equation*}
is a control flow graph for $\A$ relative to $\sigma$. Implementing the construction given in Theorem \ref{thm-graph}, we obtain the following infinite control flow graph $G_\omega$ for $\A_\omega$:
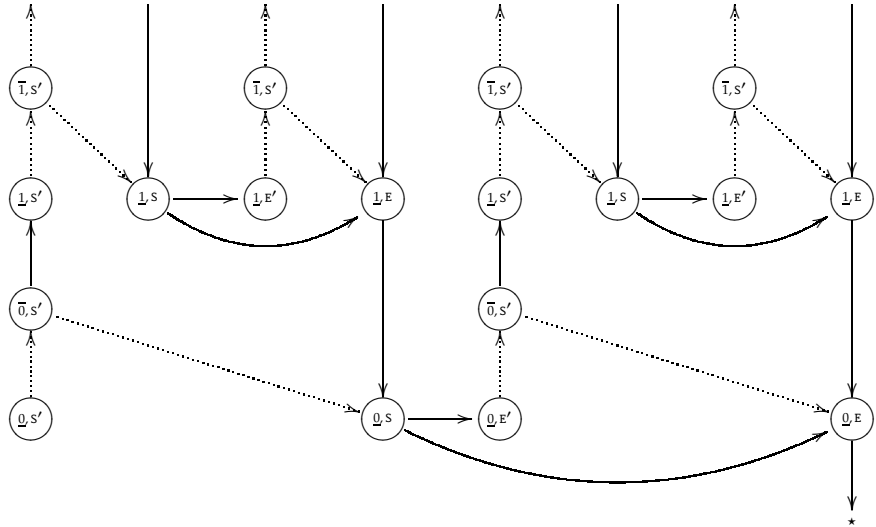
\begin{equation*}\tiny\xymatrix{ & \ar[dd] & & \ar[dd] & & \ar[dd] & & \ar[dd] \\
*++[o][F]{\tp{1},\css}\ar@{..>}[u] \ar@{..>}[rd] &  & *++[o][F]{\tp{1},\css}\ar@{..>}[u]\ar@{..>}[rd] &  & *++[o][F]{\tp{1},\css}\ar@{..>}[u]\ar@{..>}[rd] &  & *++[o][F]{\tp{1},\css}\ar@{..>}[u]\ar@{..>}[rd] &  \\ 
*++[o][F]{\bt{1},\css}\ar@{..>}[u] & *++[o][F]{\bt{1},\cs}\ar@/_1.5pc/^{}[rr]\ar[r] & *++[o][F]{\bt{1},\ce'}\ar@{..>}[u] & *++[o][F]{\bt{1},\ce}\ar[dd] & *++[o][F]{\bt{1},\css}\ar@{..>}[u] & *++[o][F]{\bt{1},\cs}\ar@/_1.5pc/^{}[rr]\ar[r] & *++[o][F]{\bt{1},\ce'}\ar@{..>}[u] & *++[o][F]{\bt{1},\ce}\ar[dd]\\
*++[o][F]{\tp{0},\css}\ar[u]\ar@{..>}[drrr] & & & & *++[o][F]{\tp{0},\css}\ar[u]\ar@{..>}[drrr] & & & \\
*++[o][F]{\bt{0},\css}\ar@{..>}[u] & & & *++[o][F]{\bt{0},\cs}\ar@/_2.0pc/^{}[rrrr]\ar[r] & *++[o][F]{\bt{0},\ces}\ar@{..>}[u] & & & *++[o][F]{\bt{0},\ce}\ar[d] \\
& & & & & & & \star }
\end{equation*}
Each transition in the algorithm $\A_\omega$ corresponds to an edge in this graph: At each node we have a current approximation $\alpha$ to our choice sequence, and decisions on which edge to follow are based on simple calculations involving $\alpha$ and $b$.

We conclude by giving some concrete examples of how the computation behaves below:

\begin{example}[$N=2$ and $b=0,1,1,\ldots$]
\label{ex-tape1}
In the following, $\alpha$ denotes $\xi_\omega(\sigma,[])$ for the current state:
\begin{equation*}
\begin{aligned}
\stac{[(\cs,0)],[]}{\Box,\Box}&\orhd_1 \stac{[(\cs,0)],[]}{1,\Box} & \mbox{$\neg(\alpha_0<\alpha_1)$}\\
&\rhd_\omega\stac{[(\cs,0),(\cs,1)],[]}{\Box,\Box} &\mbox{$1\geq |\sigma|$}\\
&\orhd_1\stac{[(\cs,0),(\cs,1)],[]}{1,\Box} &b_{\alpha_1}=b_1=1\\
&\orhd_2\stac{[(\cs,0),(\cs,1)],[]}{1,2} & \mbox{$1<|\sigma|$ and $b_1=b_2=1$}\\
&\rhd_\omega\stac{[(\cs,0),(\ce_2,1)],[]}{1,2} & \mbox{$b_1=1\wedge 1\leq 2\to b_2=1$}\\
&\rhd_\omega\stac{[(\cs,0)],[1]}{1,2} & \\
&\rhd_\omega\stac{[(\ce_2,0)],[1]}{1,2} &\mbox{$b_0=1\wedge 0\leq 1\to b_1=1$} \\
&\rhd_\omega\stac{[],[0,1]}{1,2} &
\end{aligned}
\end{equation*}
We have $b_{\alpha_1}=b_1=1$ and so $b_1=b_2=1$, and the algorithm traces the following path through $G_\omega$:
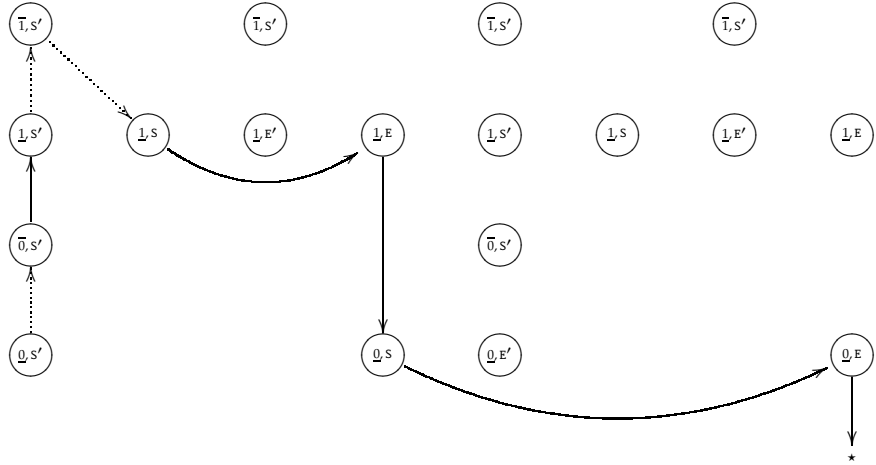
\begin{equation*}\tiny\xymatrix{*++[o][F]{\tp{1},\css} \ar@{..>}[rd] &  & *++[o][F]{\tp{1},\css} &  & *++[o][F]{\tp{1},\css} &  & *++[o][F]{\tp{1},\css} &  \\ 
*++[o][F]{\bt{1},\css}\ar@{..>}[u] & *++[o][F]{\bt{1},\cs}\ar@/_1.5pc/^{}[rr] & *++[o][F]{\bt{1},\ce'} & *++[o][F]{\bt{1},\ce}\ar[dd] & *++[o][F]{\bt{1},\css} & *++[o][F]{\bt{1},\cs} & *++[o][F]{\bt{1},\ce'} & *++[o][F]{\bt{1},\ce} \\
*++[o][F]{\tp{0},\css}\ar[u] & & & & *++[o][F]{\tp{0},\css} & & & \\
*++[o][F]{\bt{0},\css}\ar@{..>}[u] & & & *++[o][F]{\bt{0},\cs}\ar@/_2.0pc/^{}[rrrr] & *++[o][F]{\bt{0},\ces} & & & *++[o][F]{\bt{0},\ce}\ar[d] \\
& & & & & & & \star }
\end{equation*}

\end{example}

\begin{example}[$N=2$ and $b=1,0,1,0,\ldots$]
\label{ex-tape2}
In the following, $\alpha$ denotes $\xi_\omega(\sigma,[])$ for the current state:
\begin{equation*}
\begin{aligned}
\stac{[(\cs,0)],[]}{\Box,\Box}&\orhd_1 \stac{[(\cs,0)],[]}{0,\Box} & \mbox{$b_{\alpha_0}=b_0=1$}\\
&\orhd_2 \stac{[(\cs,0)],[]}{0,1} & \mbox{$0<|\sigma|$ and $b_1=0$}\\
&\rhd_\omega \stac{[(\ce_1,1)],[]}{\Box,\Box} & \mbox{$b_0=1\wedge b_1=0$}\\
&\orhd_1 \stac{[(\ce_1,1)],[]}{1,\Box} & \mbox{$\neg(\alpha_0<\alpha_1)$}\\
&\rhd_\omega \stac{[(\ce_1,1),(\cs,2)],[]}{\Box,\Box} & \mbox{$\neg(1<|\sigma|)$}\\
&\orhd_1 \stac{[(\ce_1,1),(\cs,2)],[]}{1,\Box} & \mbox{$b_{\alpha_1}=b_2=1$}\\
&\orhd_2 \stac{[(\ce_1,1),(\cs,2)],[]}{1,3} & \mbox{$1<|\sigma|$ and $b_3=0$}\\
&\rhd_\omega \stac{[(\ce_1,1),(\ce_1,3)],[]}{\Box,\Box} & \mbox{$b_2=1\wedge 2\leq 3\wedge b_3=0$}\\
&\orhd_1 \stac{[(\ce_1,1),(\ce_1,3)],[]}{1,\Box} & \\
&\orhd_2 \stac{[(\ce_1,1),(\ce_1,3)],[]}{1,4} & \mbox{$1<|\sigma|$}\\
&\rhd_\omega \stac{[(\ce_1,1)],[3]}{1,4} & \\
&\rhd_\omega \stac{[],[1,3]}{1,4} & \\
\end{aligned}
\end{equation*}
We have $b_{\alpha_1}=b_3=0$ and so $b_1=b_3=0$, and the algorithm traces the following path through $G_\omega$:

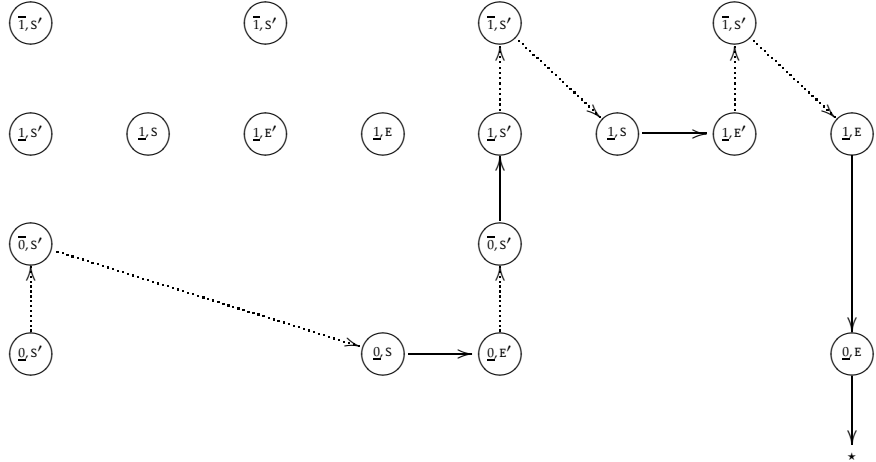
\begin{equation*}\tiny\xymatrix{ &  & &  & &  & &  \\
*++[o][F]{\tp{1},\css} &  & *++[o][F]{\tp{1},\css} &  & *++[o][F]{\tp{1},\css}\ar@{..>}[rd] &  & *++[o][F]{\tp{1},\css}\ar@{..>}[rd] &  \\ 
*++[o][F]{\bt{1},\css} & *++[o][F]{\bt{1},\cs} & *++[o][F]{\bt{1},\ce'} & *++[o][F]{\bt{1},\ce} & *++[o][F]{\bt{1},\css}\ar@{..>}[u] & *++[o][F]{\bt{1},\cs}\ar[r] & *++[o][F]{\bt{1},\ce'}\ar@{..>}[u] & *++[o][F]{\bt{1},\ce}\ar[dd]\\
*++[o][F]{\tp{0},\css}\ar@{..>}[drrr] & & & & *++[o][F]{\tp{0},\css}\ar[u] & & & \\
*++[o][F]{\bt{0},\css}\ar@{..>}[u] & & & *++[o][F]{\bt{0},\cs}\ar[r] & *++[o][F]{\bt{0},\ces}\ar@{..>}[u] & & & *++[o][F]{\bt{0},\ce}\ar[d] \\
& & & & & & & \star }
\end{equation*}

\end{example}

Note that the second example demonstrates that our program is not just a blind search: The smallest initial segment of $b$ which contains a constant subsequence is that of length $3$ i.e. $b_0=b_2=1$, but our algorithm examines the initial segment of length $4$.

In fact, our algorithm initially tries to find a consecutive subsequence $b_i=b_{i+1}=\ldots=b_{i+N-1}=1$. If there is some $i\leq j<i+N$ with $b_j=0$ then it proposes $b_j=b_{j+1}=\ldots=b_{j+N-1}=1$, and so on. At some point, it either finds a consecutive sequence of constant value $1$, or stops after $N$ failures with a (not necessarily consecutive) subsequence of constant value $0$. This is essentially the same as the program extracted in \cite{Seisenberger(2008.0)}, and demonstrates that proof theoretic methods for extracting programs from dependent choice yield non-trivial algorithms.

%%%%%%%%%%%%%%%%%%%%%%%%%%%%%%%%%%%%%%%%%%%%%%%%%
%%%%%%%%%%%%%%%%%%%%%%%%%%%%%%%%%%%%%%%%%%%%%%%%%
\section{Concluding remarks}
%%%%%%%%%%%%%%%%%%%%%%%%%%%%%%%%%%%%%%%%%%%%%%%%%
%%%%%%%%%%%%%%%%%%%%%%%%%%%%%%%%%%%%%%%%%%%%%%%%%
\label{sec-conc}

As we have emphasised from the start, we consider this paper more than just a novel computational interpretation of the axiom of choice: It is a first step towards a new way of looking at program extraction, which emphasises the high level description of programs as algorithms. As such, there are many directions for future research, and we conclude by mentioning a few of these.

The most immediate application of the ideas presented here would be to study some non-trivial mathematical theorems using sequential algorithms. An obvious candidate would be the infinite Ramsey theorem for pairs, the computational content of which has already been studied in \cite{Kreuzer(2009.0)} and \cite{OliPow(2015.0)} using variants of bar recursion. Here, the construction of Section 5 would essentially play the role of building approximations to the so-called Erd\H{o}s/Rado tree used to prove the theorem in \cite{ErdRad(1984.0)}, and our algorithmic framework would hopefully give a clear description of how the resulting approximations to pairwise monochromatic subsets are built in a sequential manner.

It would be interesting to explore refinements and extensions of the construction developed in Section \ref{sec-dc}. For the case of countable (i.e. non-dependent) choice, a non-sequential variant of our algorithm analogous to the symmetric variant of bar recursion studied in \cite{OliPow(2017.0)} would be more efficient, and would also have the advantage that it could be extended to choice over arbitrary discrete structures, as in Section 6 of \cite{OliPow(2017.0)}. More interesting would be the prospect of generalising our construction to choice over non-discrete spaces by working with some countable basis for that space. Model theoretic issues of this kind, which would be tricky to represent in a formal theory like System T, could perhaps be more smoothly managed in our more informal setting.

Finally, a much broader and more challenging problem would be to explore the connection, already hinted at in the introduction, between our algorithmic construction of choice sequences and techniques found in dynamical algebra. A recurring theme in the latter is the development of algorithms which mimic the use of Zorn's lemma. These algorithms are typically devised without the rigorous application of proof theoretic techniques, and are often given illuminating diagrammatic descriptions, as in e.g. \cite[Section 6]{Schuster(2013.0)} and \cite[Section 3]{Yengui(2008.0)}. It would be fascinating to apply the techniques of Section \ref{sec-dc} to concrete case studies in abstract algebra, particularly the graphical representation given in Theorem \ref{thm-graph}, to see to what extent they resemble existing algorithms.

%\label{sec-nci}
%\label{sec-alg}
%\label{sec-thmalg}
%\label{sec-dc}
%\label{sec-tape}
%\label{sec-conc}

\bibliographystyle{plain}
\bibliography{/home/thomas/Dropbox/tp}

\begin{thebibliography}{10}

\bibitem{BergBS(2002.0)}
U.~Berger, W.~Buchholz, and H.~Schwichtenberg.
\newblock Refined program extraction from classical proofs.
\newblock {\em Annals of Pure and Applied Logic}, 114:3--25, 2002.

\bibitem{BerCur(1982.0)}
G.~Berry and Curien P.-L.
\newblock Sequential algorithms on concrete data structures.
\newblock {\em Theoretical Computer Science}, 20(3):265--321, 1982.

\bibitem{CosLomRoy(2001.0)}
M.~Coste, H.~Lombardi, and M.-F. Roy.
\newblock Dynamical method in algebra: {E}ffective {N}ullstellens\"atze.
\newblock {\em Annals of Pure and Applied Logic}, 111(3):203--256, 2001.

\bibitem{ErdRad(1984.0)}
P.~Erd\H{o}s, A.~Hajnal, A.~M\'{a}t\'{e}, and R.~Rado.
\newblock {\em Combinatorial Set Theory: Partition Relations for Cardinals},
  volume 106 of {\em Studies in Logic and the Foundations of Mathematics}.
\newblock North-Holland Publishing Company, 1984.

\bibitem{Goedel(1958.0)}
K.~G{\"o}del.
\newblock {\"U}ber eine bisher noch nicht ben{\"u}tzte {E}rweiterung des
  finiten {S}tandpunktes.
\newblock {\em dialectica}, 12:280--287, 1958.

\bibitem{Gurevich(2000.0)}
Y.~Gurevich.
\newblock Sequential abstract-state machines capture sequential algorithms.
\newblock {\em ACM Transactions on Computational Logic (TOCL)}, 1:77--111,
  2000.

\bibitem{HylOng(2000.0)}
J.~M.~E. Hyland and C.~H.~L. Ong.
\newblock {On full abstraction for PCF: I, II and III}.
\newblock {\em Information and Computation}, 163:285--408, 2000.

\bibitem{Kohlenbach(2008.0)}
U.~Kohlenbach.
\newblock {\em Applied Proof Theory: Proof Interpretations and their Use in
  Mathematics}.
\newblock Monographs in Mathematics. Springer, 2008.

\bibitem{KohSaf(2014.0)}
U.~Kohlenbach and P.~Safarik.
\newblock Fluctuations, effective learnability and metastability in analysis.
\newblock {\em Annals of Pure and Applied Logic}, 165:266--304, 2014.

\bibitem{Kreisel(1951.0)}
G.~Kreisel.
\newblock On the interpretation of non-finitist proofs, {P}art {I}.
\newblock {\em Journal of Symbolic Logic}, 16:241--267, 1951.

\bibitem{Kreisel(1952.0)}
G.~Kreisel.
\newblock On the interpretation of non-finitist proofs, {P}art {II}:
  Interpretation of number theory.
\newblock {\em Journal of Symbolic Logic}, 17:43--58, 1952.

\bibitem{Kreisel(1959.0)}
G.~Kreisel.
\newblock Interpretation of analysis by means of functionals of finite type.
\newblock In A.~Heyting, editor, {\em Constructivity in Mathematics}, pages
  101--128. North-Holland, Amsterdam, 1959.

\bibitem{Kreuzer(2009.0)}
A.~Kreuzer.
\newblock Der {S}atz von {R}amsey f{\"u}r {P}aare und beweisbar rekursive
  {F}unktionen.
\newblock Diploma thesis, February 2009.

\bibitem{Krivine(2003.0)}
J.-L. Krivine.
\newblock Dependent choice, ‘quote’ and the clock.
\newblock {\em Theoretical Computer Science}, 308(1--3):259--276, 2003.

\bibitem{Krivine(2009.0)}
J.-L. Krivine.
\newblock Realizability in classical logic.
\newblock In {\em Interactive Models of Computation and Program Behaviour},
  volume~27 of {\em Panoramas et Synth\`{e}ses}, pages 197--229.
  Soci\'{e}t\'{e} Math\'{e}matique de France, 2009.

\bibitem{Murthy(1990.0)}
C.~R. Murthy.
\newblock {\em Extracting Constructive Content from Classical Proofs}.
\newblock PhD thesis, Ithaca, New York, 1990.

\bibitem{OliPow(2012.2)}
P.~Oliva and T.~Powell.
\newblock On {S}pector's bar recursion.
\newblock {\em Mathematical Logic Quarterly}, 58:356--365, 2012.

\bibitem{OliPow(2015.1)}
P.~Oliva and T.~Powell.
\newblock A constructive interpretation of {R}amsey's theorem via the product
  of selection functions.
\newblock {\em Mathematical Structures in Computer Science}, 25(8):1755--1778,
  2015.

\bibitem{OliPow(2015.0)}
P.~Oliva and T.~Powell.
\newblock A game-theoretic computational interpretation of proofs in classical
  analysis.
\newblock In {\em Gentzen's Centenary: The Quest for Consistency}, pages
  501--532. Springer, 2015.

\bibitem{OliPow(2017.0)}
P.~Oliva and T.~Powell.
\newblock Spector bar recursion over finite partial functions.
\newblock {\em Annals of Pure and Applied Logic}, 168(5):887--921, 2017.

\bibitem{Powell(2016.0)}
T.~Powell.
\newblock G{\"o}del's functional interpretation and the concept of learning.
\newblock In {\em Proceedings of Logic in Computer Science (LICS 2016)}, pages
  136--145. ACM, 2016.

\bibitem{Powell(2018.2)}
T.~Powell.
\newblock Computational interpretations of classical reasoning: {F}rom the
  epsilon calculus to stateful programs.
\newblock To appear in: Centrone, S., Negri, S., Sarikaya, D. and Schuster, P.
  eds., \emph{Mathesis Universalis, Computability and Proof}, Synthesis
  Library, Springer., 2018.

\bibitem{Powell(2018.1)}
T.~Powell.
\newblock A functional interpretation with state.
\newblock In {\em Proceedings of Logic in Computer Science (LICS 2018)}. ACM,
  2018.

\bibitem{Raffalli(2004.0)}
C.~Raffalli.
\newblock Getting results from programs extracted from classical proofs.
\newblock {\em Theoretical Computer Science}, 323:49--70, 2004.

\bibitem{Schuster(2013.0)}
P.~Schuster.
\newblock Induction in algebra: A first case study.
\newblock {\em Logical Methods in Computer Science}, 9(3:20):1--19, 2013.

\bibitem{Seisenberger(2008.0)}
M.~Seisenberger.
\newblock Programs from proofs using classical dependent choice.
\newblock {\em Annals of Pure and Applied Logic}, 153(1-3):97--110, 2008.

\bibitem{Spector(1962.0)}
C.~Spector.
\newblock Provably recursive functionals of analysis: a consistency proof of
  analysis by an extension of principles in current intuitionistic mathematics.
\newblock In F.~D.~E. Dekker, editor, {\em Recursive Function Theory: Proc.
  Symposia in Pure Mathematics}, volume~5, pages 1--27. American Mathematical
  Society, Providence, Rhode Island, 1962.

\bibitem{Troelstra(1973.0)}
A.~S. Troelstra.
\newblock {\em Metamathematical Investigation of Intuitionistic Arithmetic and
  Analysis}, volume 344 of {\em Lecture Notes in Mathematics}.
\newblock Springer, Berlin, 1973.

\bibitem{Yengui(2008.0)}
I.~Yengui.
\newblock Making the use of maximal ideals constructive.
\newblock {\em Theoretical Computer Science}, 392:174--178, 2008.

\end{thebibliography}
\end{document}